\documentclass[12pt,letterpaper]{article}

\usepackage{graphicx}
\usepackage{subcaption}
\usepackage{float}

\usepackage{amsmath,amssymb,amsfonts,amsthm,bbm}
\usepackage{mathtools}   
\usepackage{stmaryrd}    
\usepackage{enumitem}   

\usepackage{authblk} 

\usepackage[titletoc,title]{appendix}  

\usepackage[normalem]{ulem}

\usepackage{kpfonts}
\usepackage[T1]{fontenc}

\usepackage[table,xcdraw,dvipsnames]{xcolor}   
\usepackage{hyperref}
\newcommand\myshade{85}
\colorlet{mylinkcolor}{YellowOrange}
\colorlet{mycitecolor}{Aquamarine}
\colorlet{myurlcolor}{violet}

\hypersetup{
linkcolor  = mylinkcolor!\myshade!black,
citecolor  = mycitecolor!\myshade!black,
urlcolor   = myurlcolor!\myshade!black,
colorlinks = true,
}

\usepackage{caption}
\usepackage{subcaption}
\usepackage{graphicx}

\usepackage[linesnumbered,ruled,vlined]{algorithm2e}

\SetCommentSty{mycommfont}

\SetKwInput{KwInput}{Input}                
\SetKwInput{KwOutput}{Output}              

\usepackage{listings}             
\usepackage{stmaryrd} 

\renewcommand{\hat}{\widehat}
\renewcommand{\tilde}{\widetilde}

\newcommand{\bfm}[1]{\ensuremath{\boldsymbol{#1}}} 
\def\D{\mathrm{d}}

\def\bbone{\mathbbm{1}} 

     \def\EE{\mathbb{E}}

     \def\NN{\mathbb{N}}
     
     \def\PP{\mathbb{P}}
     
     \def\RR{\mathbb{R}}

\def\bx{\bfm x}     
     
\def\bz{\bfm z}

\def\calH{{\cal  H}}

\def\calN{{\cal  N}} 
 \def\cO{{\cal  O}}

\def\calS{{\cal  S}} 
\def\calT{{\cal  T}} \def\cT{{\cal  T}}

\def\calX{{\cal  X}}

\newcommand{\bfsym}[1]{\ensuremath{\boldsymbol{#1}}}

 \def\balpha{\bfsym \alpha}
 \def\bbeta{\bfsym \beta}

\providecommand{\abs}[1]{\left\lvert#1\right\rvert}

\providecommand{\paren}[1]{\left( #1 \right)}
\providecommand{\brackets}[1]{\left[ #1 \right]}

\providecommand{\braces}[1]{\left\{ #1 \right\}}
\providecommand{\floors}[1]{\lfloor #1 \rfloor}

\providecommand{\defeq}{:=}

\usepackage{mathtools}
\DeclarePairedDelimiterX{\infdivx}[2]{(}{)}{%
  #1 \; \delimsize\| \; #2%
}

\DeclareMathOperator{\argmax}{argmax}

\DeclareMathOperator{\Reg}{Reg}
\DeclareMathOperator{\reg}{reg}
\DeclareMathOperator{\rev}{rev}
\DeclareMathOperator{\Unif}{Unif}
\DeclareMathOperator{\DNN}{DNN}
\DeclareMathOperator{\TDNN}{TDNN}
\newcommand{\ignore}[1]{}{}

\newtheorem{definition}{Definition}
\newtheorem{assumption}[definition]{Assumption}
\newtheorem{lemma}[definition]{Lemma}

\newtheorem{theorem}[definition]{Theorem}

\theoremstyle{definition}
\newtheorem{remark}{Remark}

\definecolor{royalpurple}{rgb}{0.47, 0.32, 0.66}
\definecolor{greenfresh}{HTML}{00897B}
\definecolor{bluefresh}{HTML}{1E88E5}
\definecolor{redfresh}{HTML}{E53935}

\definecolor{royalpurple}{rgb}{0.47, 0.32, 0.66}

\def\beq{\begin{equation}}
\def\eeq{\end{equation}}

\def\bet{\begin{theorem}}
\def\eet{\end{theorem}}

\def\bel{\begin{lemma}}
\def\eel{\end{lemma}}

\def\eps{\varepsilon}

\def\cond{\;|\;}

\usepackage{setspace}
\usepackage{etoolbox}

\usepackage[framemethod=TikZ]{mdframed} 

\usepackage{fancybox}

\usepackage[top=1in, bottom=1in, left=1in, right=1in]{geometry}

\usepackage{setspace}

\usepackage[authoryear]{natbib}  
\newcommand{\mybibsty}{chicago}
\newcommand{\mybib}{bib/pricing,bib/bandits,bib/nonpara}

\usepackage{graphicx}
\usepackage{multirow}
\graphicspath{{figs/}}

\def\TTL{Dynamic Contextual Pricing with \\ Doubly Non-Parametric Random Utility Models}

\begin{document}
%
%

\newcommand{\blind}{1}

\if1\blind
{
\title{\bf \TTL}
\author{
Elynn Chen$^\flat$ \hspace{2ex}
Xi Chen$^\natural$\thanks{Corresponding author.} \hspace{2ex}
Lan Gao$^\sharp$ \hspace{2ex}
Jiayu Li$^\dag$ \\ \normalsize
\smallskip
$^{\flat,\natural,\dag}$Leonard N. Stern School of Business, New York University. \\ \normalsize
\smallskip
$^{\sharp}$ Haslam College of Business, University of Tennessee Knoxville.
}
\date{June 8, 2024}
\maketitle
} \fi

\if0\blind
{
\bigskip
\bigskip
\bigskip
\begin{center}
{\LARGE\bf \TTL}
\end{center}
\date{}
\medskip
} \fi

%
%

\setstretch{1.28}
\begin{abstract}
In the evolving landscape of digital commerce, adaptive dynamic pricing strategies are essential for gaining a competitive edge. This paper introduces novel {\em doubly nonparametric random utility models} that eschew traditional parametric assumptions used in estimating consumer demand's mean utility function and noise distribution. Existing nonparametric methods like multi-scale {\em Distributional Nearest Neighbors (DNN and TDNN)}, initially designed for offline regression, face challenges in dynamic online pricing due to design limitations, such as the indirect observability of utility-related variables and the absence of uniform convergence guarantees. We address these challenges with innovative population equations that facilitate nonparametric estimation within decision-making frameworks and establish new analytical results on the uniform convergence rates of DNN and TDNN, enhancing their applicability in dynamic environments.

Our theoretical analysis confirms that the statistical learning rates for the mean utility function and noise distribution are minimax optimal. We also derive a regret bound that illustrates the critical interaction between model dimensionality and noise distribution smoothness, deepening our understanding of dynamic pricing under varied market conditions. These contributions offer substantial theoretical insights and practical tools for implementing effective, data-driven pricing strategies, advancing the theoretical framework of pricing models and providing robust methodologies for navigating the complexities of modern markets.

\end{abstract}

\setstretch{1.9}

\section{Introduction}  \label{sec:intro}

\noindent The key in dynamic pricing is to model and estimate consumer demand function accurately. In practice, consumer demand is influenced not only by the price of a product but also by specific customer and product information, such as browsing history, zip codes, and product features. In digital sales environments, the economic model of random utility augmented with a wealth of sales data can offer crucial insights into consumer behavior, particularly consumers' preferences and their reactions to different prices under varying contextual conditions. 
By implementing contextual pricing with random utility models, firms can not only gain deeper insights into the market, but also set dynamic prices tailored to individual customer profiles and product characteristics, potentially maximizing long-term revenue and gaining a competitive edge in the market.

Literature in dynamic contextual pricing often assumes a parametric form for either the mean utility (see, e.g., \cite{NEURIPS2019,fan2022policy,luo2024distribution}), which represents the average preference of consumers, or the market noise (see, e.g., \cite{javanmard2019dynamic,cohen2020feature,xu2021logarithmic}), which captures the unobservable influences on demand. However, relying on parametric estimation poses significant disadvantages, primarily the risk of model misspecification. Parametric models make strong assumptions about the underlying data structure, and if these assumptions do not hold, the model's predictions can be inaccurate, leading to suboptimal pricing strategies.

To address these challenges, our paper explores {\em contextual pricing with doubly nonparametric random utility models}. These models estimate both the mean utility function and the market noise's distribution in a nonparametric manner, avoiding the pitfalls of model misspecification. The learning and decision process in these models is particularly challenging as the quantities of interest, including consumers' true utility and the underlying market noise, are not directly observable. Sophisticated non-parametric estimation techniques designed for offline learning are not directly applicable, and they often lack the necessary theoretical guarantees for dynamic environments. We propose an innovative online approach to estimate and optimize under these models, overcoming the inherent challenges and offering a more robust framework for dynamic pricing.

\subsection{Doubly Nonparametric Random Utility Models} \label{sec:model}

\noindent The proposed {\em doubly nonparametric random utility model} are characterized by \eqref{eqn:nonpar-utility} and \eqref{eqn:y-logistic} without any parametric characterizations on $\mathring{\mu}(\bx_t)$ or $F(\cdot)$. It operates within a broad economic framework of random utility models.
In this paper, we consider the setting in which products are sold individually, with each sale resulting in a binary outcome indicating success or failure. At each decision point $t\in[T]$, a $d$-dimensional contextual covariate $\bx_t$ from $\RR^d$ is independently and identically drawn from a fixed, but a priori unknown, distribution. 
The expected market value of the product at time \( t \) is determined by the observed contextual covariate \( \bx_t \).
A {\em general random utility model} for the market value $v_t$ of the product is given by
\begin{equation} \label{eqn:nonpar-utility}
v_t = \mathring{\mu}(\bx_t) + \eps_t,
\end{equation}
where $\mathring{\mu}(\cdot)$ is the {\em unknown function} of the mean utility, and $\eps_t$ are i.i.d.~noises following an {\em unknown distribution} $F(\cdot)$ with $\EE [\eps_t] = 0$.
We assume that the market value $v_t$ is within an interval $[B_v, B - B_v]$ with constants $0 < B_v < B$. 
Given a posted price of $p_t$ for the product at time $t$, we observe $y_t\defeq\bbone(v_t\geq p_t)$ that indicates whether the a sale occurs ($y_1=1$) or not ($y_1=0$). 
The model is equivalent to the probabilistic model:
\begin{equation} \label{eqn:y-logistic}
    y_t = \begin{cases}
        0, &\;\text{with probability}\quad F(p_t-\mathring{\mu}(\bx_t)), \\
        1, &\;\text{with probability}\quad 1- F(p_t-\mathring{\mu}(\bx_t)).
    \end{cases}
\end{equation} 
Given a posted price $p_t$, the expected revenue at time $t$ conditioned on $\bx_t$ is
\begin{equation} \label{eqn:rev_t}
    \rev_t(p_t)\defeq p_t\cdot\paren{1 - F(p_t-\mathring{\mu}(\bx_t))}.
\end{equation}
The {\em doubly nonparametric contextual pricing} operates within a general economic framework of random utility models, which diverge significantly from approaches that directly model the revenue function.
For example, \cite{chen2021nonparametric} directly models the unknown expected revenue as $\rev_t(p_t)=p_t\cdot d(\bx_t, p_t)$ with $d(\bx_t, p_t)\in[0,1]$ representing the personalized demand function. 
\cite{bu2022context} further assumes that $d(\bx_t, p_t)$ is separable in terms of $p_t$ and $\bx_t$. 
In contrast, we utilize the random utility model where the unknown expected revenue is given by \eqref{eqn:rev_t} and the demand function is assumed to be $d(\bx_t, p_t) = 1 - F(p_t - \mathring{\mu}(\bx_t))$.  The finer granularity in demand decomposition allows for separate analyses of the mean utility function and noise distribution, thus providing more precise control over the estimation errors of $F$ and $\mathring{\mu}$ and improving their statistical guarantees.

In addition, the structure of the random utility model facilitates an advanced analysis of how regret is influenced quadratically by the difference between optimal and estimated prices, which are closely tied to the estimation accuracy of $F$ and $\mathring{\mu}$. 
This connection leads to improved regret bounds under various scenarios. For further details on these relationships and their implications, refer to the discussions in Remark \ref{rmk:why lower regret} in Section \ref{sec:regret}. 

The optimal price $p_t^*$ is the price that maximizes $\rev_t(p_t)$. 
For a policy $\pi$ that sets price $p_t$ for $t$ , its regret over the time horizon of $T$ is defined as
\begin{equation} 
\begin{aligned}
\Reg(T;\pi) &= \sum_{t = 1}^T \EE\brackets{ p_t^* \mathbbm{1}(v_t \geq p_t^*)-p_t\mathbbm{1}(v_t\geq p_t)} \\
&\equiv \sum_{t = 1}^T \brackets{\rev_t(p_t^*) - \rev_t(p_t)}.
\end{aligned}
\end{equation}
The goal of this paper is to design a pricing policy that minimizes $\Reg(T;\pi)$, or equivalently, maximizes the collected expected revenue $\sum_{t = 1}^T\rev_t(p_t)$.

\subsection{Main Contributions} \label{sec:novelty}

\noindent This paper distincts from existing literature by adopting a {\em fully nonparametric} approach to model both the mean utility function, $\mathring{\mu}(\cdot)$, and the noise distribution, $F(\cdot)$, in the context of dynamic pricing with random utility. This approach eliminates the risk of model misspecification and enhances the flexibility and accuracy of our pricing strategy, which is particularly beneficial in complex and evolving market environments. 
We employ state-of-the-art nonparametric methods, notably the Distributional Nearest Neighbors (DNN) and the Two-scale Distributed Nearest Neighbor (TDNN) techniques, which are known for their simplicity in implementation and minimal assumptions. These estimators not only provide the flexibility needed in dynamic pricing but also achieve optimal convergence rates, improving as we scale up the model's complexity and data fidelity.

However, existing DNN and TDNN methods, traditionally used for offline regression, are not directly applicable to the dynamic and online nature of pricing decisions due to their design limitations. Key challenges include the indirect observability of utility-related variables and the absence of robust convergence guarantees in such a dynamic context. To overcome these limitations, we employed two new population equations, \eqref{eqn:general} and \eqref{noise-dist}, which enable the nonparametric estimation of both $\mathring{\mu}(\cdot)$ and $F(\cdot)$ using only binary response data $y_t$. This innovative approach bypasses the need for direct measurements of variables like $v_t$ or $\epsilon_t$, enhancing the robustness and accuracy of our estimations, and thereby substantially reducing the regret associated with pricing decisions.

Our theoretical developments include a new uniform convergence rate in Theorem \ref{le-DNN-TDNN} for the DNN and TDNN estimators. 
Employing the elegant U-Statistic representation \citep{lee2019u} of DNN and TDNN, Theorem \ref{le-DNN-TDNN} establishes that the DNN and TDNN estimators achieve the {\em minimax optimal rate} of uniform convergence for nonparametric regression under the second-order smoothness condition and the fourth-order smoothness condition, respectively \citep{Stone1982}.

Moreover, our theoretical analyses reveal new insights into the interdependencies between the estimation accuracy for $\mathring{\mu}(\cdot)$ and $F(\cdot)$, and their impact on the overall performance of the pricing strategy within the {\em doubly nonparametric setting}. Firstly, our analysis shows that the regret bounds for our {\em doubly nonparametric} dynamic pricing model intricately depend on the trade-offs between the dimensionality of contextual variables ($d$) and the smoothness properties of the noise distribution ($m$), as detailed in the last row of Table \ref{tab:compare-non-par}. Secondly, the regret bounds for our methods outperform those of existing nonparametric models \citep{chen2021nonparametric,bu2022context}, offering more precise control over varying smoothness levels of $F$ and $\mathring{\mu}$, and achieving enhanced performance in comparable scenarios.

Specifically, when the revenue function exhibits $2$nd-order smoothness, (i.e. $k_a=k_b=2$, $m=2$), our approach using the DNN estimator achieves a regret bound of $\tilde{\cO}(T^{\frac{5}{7} \lor \frac{d + 4}{d + 8}})$, , which is tighter than those achieved by other methods. Similarly, when the revenue function exhibits $4$th-order smoothness,  (i.e. $k_a=k_b=4$, $m=4$), our method utilizing the TDNN estimator achieves the regret bound of $\tilde{\cO}(T^{\frac{3}{5} \lor \frac{d + 8}{d + 16}})$, again tighter than those offered by other models.

\begin{table}[htpb!]
\centering
\resizebox{1\textwidth}{!}{%
\begin{tabular}{c|c|c|l}
\hline
Paper & Revenue & Regret & Key Assumptions \\ \hline
\cite{slivkins2011contextual} & $f(p, \bx)$ & $\widetilde{\cO} (T^{\frac{d + 2}{d + 3}}) $ & $f$: Lipschitz continuous \\ \hline
\cite{chen2021nonparametric} & $f(p, \bx)$ & $\widetilde{\cO} (T^{\frac{d + 2}{d + 4}}) $ & \begin{tabular}[c]{@{}l@{}}$f$: Lipschitz continuous;\\ 
 $f_B(p)$: smooth and concave\end{tabular}   
\\ \hline
\cite{bu2022context} & $p(a(\bx) + b(p))$ & $\widetilde{\cO}(T^{\frac{3}{4}\lor  \frac{d + k_a}{d + 2 k_a}} )$  & \begin{tabular}[c]{@{}l@{}}$a$: $k_a$th-order smooth;\\ 
 $b$: $k_b$th-order smooth \end{tabular} \\ \hline
 \multirow{2}{*}{This work} & \multirow{2}{*}{$p(1 - F(p - \mathring{\mu}(\bx)) )$} & DNN: $\tilde{\cO}(T^{\frac{2m+1}{4m-1} \lor \frac{d + 4}{d + 8}})$ & \begin{tabular}[c]{@{}l@{}}$F$: $m$th-order smooth;\\ 
 $\mathring{\mu}$: 2nd-order smooth \end{tabular}   \\ \cline{3-4} 
 &  & TDNN: $\tilde{\cO}(T^{\frac{2m+1}{4m-1} \lor \frac{d + 8}{d + 16} \lor \frac{7}{13}})$ & \begin{tabular}[c]{@{}l@{}}$F$: $m$th-order smooth;\\ 
 $\mathring{\mu}$: $4$th-order smooth \end{tabular}   \\ \hline
\end{tabular}%
}
\caption{Comparison of the main results in this work and existing non-parametric dynamic pricing}
\label{tab:compare-non-par}
\end{table}

\subsection{Literature and Organization} \label{sec:literature}

\noindent We explore the extensive literature on contextual dynamic pricing, focusing on the underlying economic models and nonparametric techniques. Generally, this literature splits into two categories: contextual dynamic pricing under random choice models and direct modeling of demand or revenue functions. The latter is often treated as a special case of bandit problems where the reward function corresponds to the revenue function. Additionally, we review related work in dynamic pricing, including recent developments that incorporate inventory control, fairness, and privacy considerations. These areas could potentially integrate with our {\em doubly nonparametric contextual pricing} approach in future research directions.

\noindent
\textbf{Contextual dynamic pricing under random choice models.}  
Recently, there has been a growing body of literature on context-based dynamic pricing, where product and customer features are taken into account when modeling the demand curve or market evaluation.
We refer the reader to recent surveys \citep{den2015dynamic,ban2021personalized,chen2022primaldual,chen2022elements,chen2023data} for a comprehensive review on this stream of literature. The most relevant literature to the present paper is those that model the demand by {\it a binary random choice model} characterized by \eqref{eqn:nonpar-utility} and \eqref{eqn:y-logistic}. 
This group of literature can be categorized by the assumptions on the two unknown functions, namely the mean utility $\mathring{\mu}(\bx_t)$ and the CDF of the noise $F(\cdot)$. 
Assuming the CDF $F(\cdot)$ is a known priori, \cite{javanmard2019dynamic,cohen2020feature,xu2021logarithmic} all impose a linear parametric model on the expected market value of the product at time $t$, i.e. $\mathring{\mu}(\bx_t) = \balpha^{\top}\bx_t$ in \eqref{eqn:nonpar-utility}, where $\balpha $ is some unknown parameter, and aim to estimate the mean utility while optimizing the revenue. 
Alternatively, 
\cite{NEURIPS2019,fan2022policy,luo2024distribution} allow the noise distribution, i.e. the CDF $F(\cdot)$ in \eqref{eqn:y-logistic}, to be unknown. \citep{fan2022policy,luo2024distribution} assume an underlying linear mean utility $\mathring{\mu}(\bx_t)$ and estimate $F$ by kernel estimation and by the discretization method \citep{kleinberg2003value}, respectively. While \cite{NEURIPS2019} assume the market value \(v_t = \exp(\balpha^T \bx_t + \eps_t)\), where \( \eps_t \) has unknown distribution and propose the DEEP-C algorithm to optimize the regret. However, they all work under the settings where the mean utility is parametrically organized.

In contrast, we consider {\it doubly nonparametric} random utility models where neither $\mathring{\mu}(\cdot)$ or $F(\cdot)$ is imposed with parametric forms which is more flexible yet more challenging than considering only one non-parametric component. 
 
\medskip
\noindent
\textbf{Nonparametric contextual bandits and dynamic pricing as a special case.} The literature on bandit problems is vast. 
We refer the reader to \cite{lattimore2020bandit,slivkins2019introduction} for comprehensive reviews. 
The most relevant literature in bandits to the present paper is nonparametric contextual bandit problems. 
The general contextual bandit problems consider estimating and optimizing over a random reward $r_t$ whose expectation is $\EE[r_t \cond \bx_t, p] = f(\bx_t, p)$ where $\bx_t$ is the contextual information and $p$ is the action taken. 
The $K$-Multi-Arm contextual bandit considers discrete action $p \in [K]$ and employs model $f(\bx_t, p) = f^{(p)}(\bx_t)$ for finite action space.
While most bandit literature studies this problem with additional parametric models on $f^{(p)}(\bx_t)$, \cite{perchet2013multi,hu2022smooth} investigate the problem with nonparametric reward feedback under a general H\"{o}lder continuous assumption.

A stream of literature in contextual dynamic pricing models the demand function $d(\bx, p)$ or the revenue function $p\cdot d(\bx, p)$ directly \citep{qiang2016dynamic,ban2021personalized,chen2021nonparametric,Wang2021Uncertainty,bu2022context}. This stream of work can be categorized into contextual bandits with a specialized reward function $f(\bx_t, p) = p \cdot d(\bx, p)$ and continuous scalar action $p$. Under the general bandit setting, \cite{slivkins2011contextual} studies nonparametric bandits assuming that $f(\bx_t, p)$ is Lipschitz and establishes the optimal regret being $\tilde{\cO}(T^{\frac{d+2}{d+3}})$ when $p$ is a scalar\footnote{\cite{slivkins2011contextual} works with a general setting where an action $p$ can be a vector.}. \cite{NEURIPS2018} exploit a special structure where the demand function is modeled as a strictly binary feedback $\mathbf{1}\{p \leq \mathring{\mu}(\bx_t)\}
$, where $\mathring{\mu}(\bx_t)$ is Lipschitz, and develops the regret bound of  $\tilde{\cO}(T^{\frac{d}{d+1}})$. 
\cite{guan2018nonparametric} apply $K$-Nearest Neighbour to nonparametric stochastic contextual bandits attaining a regret of { $\tilde{\cO}(T^{\frac{d+1}{d+2}})$}.
Under the dynamic pricing setting, \cite{chen2021nonparametric} estimate $f(\bx, p) = p \cdot d(\bx, p)$ directly.
However, they do not impose any functional structure that is specific to the dynamic pricing problem. Assuming that the revenue function is Lipschitz continuous and locally concave, the authors prove that the optimal regret is $\tilde{\cO}(T^{\frac{d+2}{d+4}})$ where $d$ is the dimension of $\bx$. \cite{nambiar2019dynamic} study the demand function \(bp + g(x)\) motivated by model misspecification and propose a random price shock (\textit{RPS}) algorithm to estimate price elasticity. 
They establish the optimal regret of order \(\tilde{\cO}(\sqrt{T})\) compared to a clairvoyant who knows the best linear approximation to \(g(x)\).
\cite{bu2022context} further explore the separable structure of the demand function in the sense that $d(\bx, p) = a(\bx) + b(p)$ where both function $a(\cdot)$ and $b(\cdot)$ are unknown. 
This paper adopts the random utility model for demand characterization and is very different from these works in terms of aggregate demand modeling, algorithm development, and regret analysis.

\medskip
\noindent
\textbf{Contextual dynamic pricing with related topics.} There has been a wide range of literature studying contextual dynamic pricing with additional interesting topics including inventory constraints, fairness, privacy protection and more. We refer the readers to some recent works that could potentially incorporate our method in future research. In terms of joint pricing and inventory control problem, \cite{chen2021jointpricing} develop learning algorithm with censored demand and approximates $f(d_t, p_t, y_t)$ nonparametrically, where $y_t$ represents inventory levels. \cite{feng2023principro} further explore the impact of price protection guarantee, where the customers who purchased a product are endowed the right to receive a refund if the seller lowers the price during the price protection period. \cite{li2023dynamic} take the inventory constraints into consideration where the inventory level B is given at the beginning and cannot be replenished. 
All those studies develop learning algorithms to estimate the underlying unknown demand functions, where our method could be seamlessly integrated. In terms of fairness, \cite{Maestre2018Reinforcement}
consider the fairness in price, they introduce \textit{Jain's index} to measures the fairness in terms of the average price across groups, while \cite{Cohen2021Fairness} estimate unknown demand function under both price fairness and time fairness. 

Additional constraints or addition conditions of related interest in dynamic pricing are explored in the following works. \cite{Chen2021Privacy} adopt the fundamental framework of differential privacy when implementing demand learning, and develop privacy-preserving dynamic pricing, which tries to maximize the retailer revenue while avoiding leakage of individual customer’s information.
\cite{keskin2017chasing} adopt a linear underlying demand model but assumes the unknown parameters can change over time. \cite{vandenBoer2022dynamic} take reference-price effects into account and learn demand uncertainty by adding the reference-price updating mechanism with unknown parameters. \cite{bastani2021meta} leverage a potentially shared structure of the demand function across multiple products and conduct a meta dynamic pricing algorithm to make pricing more efficiently.

\medskip
\noindent
\textit{Organization.}
Section \ref{sec:estimation} presents the estimation and decision framework. 
Section \ref{sec:theory} establishes theoretical guarantees on statistical convergence of learning as well as regret upper bound on decision. 
Section \ref{sec:simul} and \ref{sec:real} conduct empirical studies on synthetic and real data sets.  

\section{Doubly Non-Parametric Learning and Decision} \label{sec:estimation}

\noindent We propose an Algorithm \ref{algo:non-par-pricing} which outputs the empirical optimal policy for contextual dynamic pricing and balances the exploration-exploitation trade-off. 
The horizon of $T$ periods is divided into episodes with increasing length. Specifically, the $k$th episode has length $n_k = 2^{k-1} n_0$, where $k \geq 1$ and $n_0$ is the minimal episode length. 
For each $k \geq 1$, the $k$th episode is further divided into exploration and commit (or exploitation) phases, which are denoted by $\cT_{k,exp}$ and $\cT_{k,com}$, respectively. 
In the exploration phase $\cT_{k,exp}$, we collect sample $(x_t, y_t)$ through randomly setting price $p_t \stackrel{d}{\sim} \Unif(0, B)$ and then estimate both $\mathring{\mu}(\cdot)$ or $F(\cdot)$ nonparametrically using the collected sample $\{ (x_t, y_t)\}_{t \in \cT_{k,exp}}$. 
In the exploitation phase $\cT_{k,com}$, we proceed to
formulate a data-driven optimal policy leveraging the estimated mean utility and noise distribution learned in the exploration phase.
The sizes of the exploration phase of uniform random pricing and the exploitation phase of greedy pricing based on estimation can be determined in theoretical analysis to minimize the total regret. 
The subsequent subsections introduce our proposed approach for the doubly non-parametric estimation of $\mathring{\mu}(\cdot)$ and $F(\cdot)$. 

\begin{algorithm}[!ht]
\caption{Doubly Non-Parametric Contextual Pricing} \label{algo:non-par-pricing}
\KwInput{Upper bound $B> 0$ for the market values, and minimal episode length $n_0$.}

Calculate the length of the $k$-th episode $n_k = 2^{k-1} n_0$, and the length of $k$-th exploration episode $n_{k, exp} < n_k$.  

Let $\calT_{k, exp} := \{n_k - n_0 +1, \cdots, n_k - n_0 + n_{k, exp} \}$ and $\calT_{k, com} :=\{ n_k - n_0 + n_{k, exp} + 1, \cdots, n_{k+1}- n_0 \}$ be the index sets of the exploration phase and the exploitation phase of the $k$-th episode, respectively. 

\tcc{******* Episodic Learning and Decision *******}

\For{\rm{each episode $k = 1, 2, \ldots$}}{

  \tcc{******* Exploration on $\calT_{k, exp}$ *******}
  \For{$t \in \calT_{k, exp}$}{
    Offer price $p_t \sim U(0, B)$ and collect data $\{(\bx_t, y_t)\}_{t \in \calT_{k, exp}}$;
  }
  
  \tcc{******* Estimation *******}

  Using samples $\{(\bx_t, y_t)\}_{t \in \calT_{k, exp}}$ to fit $\hat{\mu}_k(\bx_t)$ for $t \in \calT_{k, exp}$ according to DNN \eqref{eqn:dnn-L-stat} or TDNN \eqref{eqn:tdnn}. 

  Update estimates of the error distribution $\hat{F}_k(u)$ by the Nadaraya-Watson kernel estimator \eqref{eqn:F-est}.

  \tcc{******* Commitment on $\calT_{k, com}$ *******}

  Update first derivative by $\hat{F}_k^{(1)}(u)$ according to \eqref{eqn:F(1)-est}. 
  
  Update estimate of $\hat{h}_k (v)$ according to \eqref{eqn:func-h-hat}. 

  \For{$t \in \calT_{k, com}$}{
    Predict $\hat{\mu}_k(\bx_t)$ according to DNN \eqref{eqn:dnn-L-stat} or TDNN \eqref{eqn:tdnn}. 
   
    Set offer price  
      \begin{equation*}
        \hat{p}_t = \min \{ \max \{\hat{h}_k (\hat{\mu}_k(\bx_t)), 0 \},  B\}, 
     \end{equation*}
   where $\hat{h}_k(\cdot)$ is obtained by \eqref{eqn:func-h-hat}. 
  }   
}

\KwOutput{Offered price $p_t$, $t \geq 1$}
\end{algorithm}

\subsection{Mean Utility Function Estimation}
\label{sec:mean-utility}
\noindent To address the challenge of not directly observing random utility $v_t$, we observe the following property for both single-product dynamic pricing. 
Specifically, if a uniform random policy $p_t \sim \Unif(0, B)$ with $B$ being the upper bound for the market value is adopted, we can recover the mean utility function in an average sense. Denote by $g_t := B y_t$. 
It can be observed that given the contextual covariate $\bx_t$, $g_t$ satisfies
\begin{equation} 
\begin{aligned}
\label{eqn:general}
\EE[g_t \,| \,\bx_t] & = \EE[B y_t \,| \,\bx_t] = \EE_{\eps_t}\brackets{ \EE_{p_t}[ B y_t \,| \,\bx_t, \eps_t ] \cond \bx_t } = B \EE_{\eps_t} \Big[\frac{ \mathring{\mu}(\bx_t) + \eps_t } {B} \, \Big| \,\bx_t \Big] = \mathring{\mu}(\bx_t).
\end{aligned}
\end{equation} 
Therefore, if the seller randomly sets i.i.d. price $p_t  {\sim} \Unif(0, B)$ for each period $t$ in the exploration phase $\cT_{k, exp}$, then the mean utility function $\mathring{\mu}(\bx_t)$ can be estimated by regressing $g_t$ on $\bx_t$ based on observations $\bz_t \defeq \{(\bx_t, g_t) \}_{t \in \cT_{k, exp}}$.

\paragraph{Distributional Nearest Neighbors.} 
To non-parametrically estimate the mean utility function $\mathring{\mu}(\cdot)$, we apply the method of Distributional Nearest Neighbors (DNN, \cite{steele2009exact,biau2010rate}) which is flexible and applicable even with a limited amount of data. 
DNN is a variation of $K$-Nearest Neighbours ($K$-NN) and automatically assigns monotonic weights to the nearest neighbors in a distributional fashion on the entire sample. 
It incorporates the idea of bagging by averaging all 1-NN estimators, each constructed from randomly subsampling  $s$ observations without replacement. Here $s$ is the subsampling scale diverging with the total sample size.

Suppose a sample $ \{(\bx_t, y_t) \}_{t \in \cT_{k, exp}}$ of size $n_{k,exp}$ is collected in the exploration phase of the $k$th episode. Since $g_t := B y_t$ satisfies \eqref{eqn:general}, we build our estimator directly with $\{ (\bx_t, g_t) \}_{t \in \cT_{k, exp}}$. 
Considering a new contextual covariate observation $\bx$, our goal is to fit a non-parametric mean function from the sample to estimate the underlying truth $\mathring{\mu}(\bx)$. We reorder the observations $\{ (\bx_t, g_t) \}_{t \in \cT_{k, exp}}$ according to $\bx_t$'s Euclidean distance to $\bx$ and denote the reordered sample as $\{(\bx_{(1)}, g_{(1)}), \ldots, (\bx_{(n_{k, exp})}, g_{(n_{k, exp})}) \}$ such that
\begin{equation} \label{new-distance-overall}
   \|\bx_{(1)}-\bx\| \le \|\bx_{(2)}-\bx\| \le \cdots \le \| \bx_{(n_{k, exp})}-\bx\|,
\end{equation}
where $\|\cdot\|$ denotes the Euclidean norm, and the ties are broken by assigning the smallest rank to the observation with the smallest natural index.
$g_{(1)}$ is defined as the $1$-NN estimator of $\mathring{\mu}(\bx)$ based on the sample $\{(\bx_t, g_t) \}_{t \in \cT_{k, exp}}$. 
For simpler notation, denote $\bz_t := (\bx_t, g_t)$ for $1 \leq t\leq n_{k, \exp}$.
For $1 \leq s \leq n_{k, exp}$, the DNN estimator $\hat\mu_k^{\DNN}(\bx; s)$ with subsampling scale $s$ for estimating $\mathring\mu(\bx)$ is formally defined as the average over all $1$-NN estimators based on the $s$-scale subsamples, that is, 
\begin{equation} \label{eqn:dnn-u-stat}
\hat\mu_k^{\DNN}(\bx; s) = {n_{k, exp} \choose s}^{-1}\underset{1\le i_1\le\cdots \\
\le i_s\le n_{k, exp}}{\sum} g_{(1)}(\bz_{i_1}, \cdots, \bz_{i_s}),
\end{equation}
where $g_{(1)}(\bz_{i_1}, \cdots, \bz_{i_s})$ is the $1$-NN estimator of $\mathring{\mu}(\bx)$ based on the subsample $(\bz_{i_1}, \cdots, \bz_{i_s})$. 
Specifically, with a abuse of notation, $g_{(1)}(\bz_{i_1}, \cdots, \bz_{i_s}) := g_{(1)}$ from the reordered sample $\{(\bx_{(1)}, g_{(1)}), \cdots, (\bx_{(s)}, g_{(s)})\}$ of the subsample $\{(\bx_{i_1}, g_{i_1}), \ldots,  (\bx_{i_s}, g_{i_s})\}$ such that
\begin{equation}
\| \bx_{(1)}-\bx\| \le \|\bx_{(2)}-\bx\| \le \cdots \le \|\bx_{(s)}-\bx\|. 
\end{equation}
The subsampling scale $s$ will be chosen to balance the trade-off between bias and variance. 
In addition, it is shown in \cite{steele2009exact} that the DNN estimator also admits an equivalent representation of a weighted nearest neighbor that
\begin{equation}\label{eqn:dnn-L-stat}
\hat\mu_k^{\DNN}(\bx; s) = {n_{k, exp} \choose s }^{-1} \sum_{i = 1}^{ n_{k, exp} - s - 1}  {n_{k, exp} - i  \choose s - 1} \, g_{(i)}, 
\end{equation}
where $\{(\bx_{(1)}, g_{(1)}), \ldots, (\bx_{(n_{k, exp})}, g_{(n_{k, exp})}) \}$ is the overall ordered sample according to the ordered distance in \eqref{new-distance-overall}.

A nice feature of the above DNN estimator is that the weights on the nearest neighbors are characterized exclusively by the total sample size $n_{k, exp}$ and the subsampling scale $s$. 
Therefore, the equivalent representation in \eqref{eqn:dnn-L-stat} can be exploited for easy implementation, while the representation in \eqref{eqn:dnn-u-stat} is crucial in theoretical analysis for the uniform convergence rate of the estimator. 

\paragraph{Two-scale Distributional Nearest Neighbour.} 
When the mean utility function and density function of covariates $\bx_t$ have a higher degree of smoothness, two-scale DNN (TDNN, \cite{demirkaya2022optimal}) can be applied to eliminate the first-order bias of the DNN estimator by aggregating two DNN estimators with distinct subsampling scales $s_1$ and $s_2$. Specifically, a TDNN estimator of the mean utility function $\mathring{\mu} (\bx)$ is given by 
\begin{equation} \label{eqn:tdnn}
\hat{\mu}_k^{\TDNN} (\bx; s_1, s_2)= \alpha_1 \hat{\mu}_k^{\DNN} (\bx;s_1) + \alpha_2 \hat{\mu}_k^{\DNN} (\bx; s_2),
\end{equation}
where $\alpha_1$ and $\alpha_2$ satisfy $\alpha_1 + \alpha_2 = 1$ and $\alpha_1 s_1^{-2/d} + \alpha_2 s_2^{-2/d} = 0$ (this equation serves to remove the first-order bias). From these two equations, we derive the coefficients $\alpha_1 = 1/(1 - (s_1/s_2)^{-2/d})$ and $\alpha_2 = - (s_1 / s_2)^{-2/d} / (1 - (s_1/s_2)^{-2/d})$. 

In general, if the mean utility function $\mathring{\mu}(\bx)$ and density function of covariates $\bx_t$ are ($2l$)-times differentiable for an integer $l \geq 2$, we can construct an $l$-scale DNN estimator to eliminate the bias up to the $(l-1)$-th order by constructing a linear combination of $l$ DNN estimators with distinct subsampling scales, utilizing a similar technique in \eqref{eqn:tdnn}. 

\subsection{Utility Noise Distribution Estimation}

\noindent To non-parametrically estimate the cumulative distribution function $ F(\cdot) $ of the random noise $\eps_t$, the key observation is that it can be expresses through a conditional expectation:
\begin{equation} \label{noise-dist}
\begin{aligned}
F(z) &= \mathbb{P} ( \eps_t \leq z ) = \EE \big[ \mathbb{P} \big(  \eps_t \leq p_t - \mathring{\mu}(\bx_t) \,|\, p_t - \mathring{\mu}(\bx_t) = z \big) \big] \\
& = \EE \big[ \mathbb{P} \big(  v_t  \leq p_t   \,|\, p_t - \mathring{\mu}(\bx_t) = z \big) \big]  \\
& = 1 - \EE [y_t | p_t - \mathring{\mu} (\bx_t) = z] .
\end{aligned}
\end{equation}
Thus, the Nadaraya-Watson kernel regression estimator can be applied to estimate $F(\cdot)$. 
Given sample $\{\bx_t, y_t, p_t\}_{t\in \calT_{k, exp}}$ collected in the exploration phase and any estimated utility function $\mu(\cdot)$,  we define
\begin{equation}\label{eqn:aux-def-a-xi}
\begin{aligned}
\hat a_k(z; \mu) & \defeq \abs{n_{k, exp} b_k}^{-1}\sum_{t \in \calT_{k, exp}} K \paren{\frac{p_t - \mu(\bx_t) - z}{b_k}} y_t, \\
\hat \xi_k(z; \mu) & \defeq \abs{n_{k, exp} b_k}^{-1}\sum_{t \in \calT_{k, exp}} K\paren{\frac{p_t - \mu(\bx_t) - z }{b_k}},
\end{aligned}
\end{equation}
where $K(\cdot)$ is a chosen $m$-th order kernel function and $b_k$ is the bandwidth that will be chosen appropriately later. 
Based on the sample $\{\bx_t, y_t, p_t\}_{t\in \calT_{k, exp}}$ and any estimated utility function $\mu(\cdot)$,  we apply the Nadaraya-Watson kernel estimator of $F(z)$ defined by 
\begin{equation} \label{eqn:F-est-any-mu}
\hat{F}_k(z; \mu) = 1 - \frac{\hat a_k(z; \mu)}{\hat \xi_k(z; \mu)}. 
\end{equation}

In Section \ref{sec:mean-utility}, we have introduced the DNN and TDNN estimators of the mean utility function. For instance, if the DNN estimator $\hat{\mu}_k^{\DNN} (\bx_t; s)$ is applied for each period $t \in \calT_{k, exp}$ based on the collected sample $\{\bx_t, y_t\}_{t\in \calT_{k, exp}}$, the estimator of the distribution $F(z)$ in \eqref{eqn:F-est-any-mu} can be written as 
\begin{equation} \label{eqn:F-est}
\hat{F}_k\paren{z; \hat{\mu}_k^{\DNN}} =  1 - \frac{\hat a_k(z; \hat{\mu}_k^{\DNN})}{\hat \xi_k(z; \hat{\mu}_k^{DNN})}. 
\end{equation}
Analogously, the Nadaraya-Watson estimator of $F(z)$ based on the TDNN estimator of the mean utility function is given by 
\begin{equation} \label{eqn:F-est-TDNN}
\hat{F}_k\paren{z; \hat{\mu}_k^{\TDNN}} = 1 - \frac{\hat a_k(z; \hat{\mu}_k^{\TDNN})}{\hat \xi_k(z; \hat{\mu}_k^{\TDNN})}. 
\end{equation}

\subsection{Optimal Decision}
\noindent In the $k$-th episode, the estimators $\hat{\mu}_k(\cdot)$ for the mean utility function $\mathring{\mu} (\cdot)$ and $\hat{F}_k\paren{z; \hat{\mu}_k}$ for the noise distribution $F(z)$, derived in the exploration phase, 
will be employed to inform data-driven optimal decision for dynamic pricing throughout the exploitation phase. 

\paragraph{Oracle Optimal Decision.}
Under the oracle setting where the true $\mathring{\mu}(\cdot)$ and $F(\cdot)$ are known, the optimal offered price $p_t^*$ can be determined according to 
\begin{equation} \label{eqn:opt-p}
p_t^* = \underset{p_t \ge 0}{\argmax}~ p_t [1 - F(p_t - \mathring{\mu}(\bx_t)].
\end{equation}
Under Assumption \ref{assump:regu-phi}, the oracle optimal price is given by 
\begin{equation} 
\label{eqn:func-h}
\begin{aligned}
&p_t^* = h\circ \mathring{\mu}(\bx_t), 
\text{where}\quad h(v) &:= v + \phi^{-1}(-v) \quad \mbox{and}\quad \phi(z) := z - \frac{ 1 - F(z) } {F^{(1)}(z)}.
\end{aligned}
\end{equation}
Here $h\circ\mathring\mu(\cdot)$ denotes the composite function $h(\mathring\mu(\cdot))$ and $F^{(1)}$ is the first-order derivative of $F$. 

\paragraph{Feasible Optimal Decision.}
In reality, we can only estimate $\mathring{\mu}(\cdot)$ and $F(\cdot)$ from collected data. 
In the $k$-th episode of Algorithm \ref{algo:non-par-pricing}, we obtain the estimated mean utility function $\hat{\mu}_k$ using either the DNN or TDNN estimator and calculate the kernel estimator $\hat{F}_k(z)$ based on \eqref{eqn:F-est} or \eqref{eqn:F-est-TDNN}.
Accordingly, the data-driven optimal offered price $\hat p_t$ for $t\in\cT_{k,com}$ in the exploitation phase is obtained by optimizing the estimated reward. That is, $\hat p_t = \underset{p_t > 0}{\argmax}~ p_t [1 - \hat{F}_k(p_t - \hat{\mu}_k(\bx_t)]$. 
As a result, it can be obtained that for any $t\in\cT_{k,com}$, 
\begin{equation}\label{eqn:opt-price-est}
\hat p_t = \hat{h}_k \circ \hat{\mu}_k(\bx_t),   
\end{equation}
where 
\begin{equation} \label{eqn:func-h-hat}
\hat{h}_k(v) = v + \hat{\phi}_k^{-1}(-v) 
\quad\text{and} \quad
\hat{\phi}_k(z) = z - \frac{ 1 - \hat{F}_k(z) } {\hat{F}_k^{(1)}(z)}.
\end{equation}
Here $\hat{F}_k^{(1)}(z)$ represents the first-order derivative of $\hat{F}_k (z)$ and is defined by
\begin{equation}\label{eqn:F(1)-est}
\hat{F}_k^{(1)}(z) = \hat{F}_k^{(1)}(z; \hat{\mu}_k) = - \frac{\hat a_k^{(1)}(z;\hat{\mu}_k) \hat \xi_k(z; \hat{\mu}_k) - \hat a_k(z; \hat{\mu}_k) \hat \xi_k^{(1)}(z; \hat{\mu}_k)}{\hat \xi_k(z; \hat{\mu}_k)^2}, 
\end{equation}
where $\hat a_k(z; \hat\mu_k)$ and $\hat\xi_k(z;  \hat\mu_k)$ are defined in \eqref{eqn:aux-def-a-xi} and their corresponding first derivatives are given by 
\begin{align*}
\hat a_k^{(1)}(z;  \hat\mu_k) & \defeq \frac{-1}{n_{k, exp} b_k^2} \sum_{t \in \calT_{k, exp}} K^{(1)}\paren{\frac{p_t -  \hat\mu_k(\bx_t) - z}{b_k}} y_t, \\
\hat \xi_k^{(1)}(z;  \hat\mu_k) & \defeq \frac{-1}{n_{k, exp} b_k^2}\sum_{t \in \calT_{k, exp}} K^{(1)}\paren{\frac{p_t -  \hat\mu_k(\bx_t) - z }{b_k}}. 
\end{align*}

\section{Theoretical Guarantees}
\label{sec:theory}

In this section, we establish the regret analysis of the proposed Algorithm \ref{algo:non-par-pricing} for doubly robust learning and decision in contextual dynamic pricing. For the theoretical analysis, we first introduce some notations and present necessary assumptions. 
For any candidate mean utility function $\mu: \mathcal{X} \mapsto \mathbb{R}$, let $\xi(z;\mu)$ be the density function of $Z = p_t-\mu(\bx_t)$ evaluated at $Z = z$
and define $F(z; {\mu}) := 1 - \EE[y_t | p_t - {\mu} (\bx_t) = z]$. 
Therefore, the underlying true noise distribution $F(z) = F(z; \mathring{\mu})$ according to \eqref{noise-dist}. 
We next introduce the assumptions about the market noise $\eps_t$, the distribution function $F(u; \mu)$, density function $\xi(u;\mu)$, and the kernel function $K(x)$. 


\begin{assumption}[Smoothness of $F(z;\mu)$ and $F^{(1)}(z;\mu)$]\label{assump:F}
The true distribution of the market noise $F(z)$ is $\ell_F$-Lipschitz. 
Let $\calS_{\eps} \defeq [-B_{\eps}, B_{\eps}]$ be the support of the market noise, it is assumed that, for any $\delta > 0$, $\sup_{\| \mu - \mathring{\mu} \|_{\infty} \leq \delta, z \in \calS_{\eps}} | F^{(1)}(z; \mu) - F^{(1)}(z) | \leq \ell_F \delta$.
Moreover, it holds that $\min_{z \in [\phi^{-1} (B_\eps - B),  \phi^{-1} (- B_{\eps}) ]} F^{(1)} (z) \geq c_F$ for some constant $c_{F} > 0$. 
\end{assumption}

\begin{assumption}[Smoothness of $\xi (\cdot; \mu)$] \label{assump:xi}
There exists an integer $m \geq 2$ and a constant $\ell_{\xi}$ such that for all $\mu \in \calN_0 := \{\mu: \|\mu  - \mathring{\mu} \|_{\infty} \leq  (  T^{-\frac{ 2m +1} {(4m-1)(d +4)}}\land T^{- \frac{1} {d+8}} ) \sqrt{d \log T} \} $, the density function $ \xi(\cdot; \mu)   \in \mathbb{C}^{(m)} $ and $\xi^{(m)} (\cdot; \mu)$ is $\ell_{\xi}$-Lipschitz continuous on $\calS_{\eps}$. 
Moreover, there exist constants $B_1 >0$ and $B_2 >0$ such that for any $\mu \in \mathcal{N}_0$,  $\max(\| \xi (\cdot; {\mu} )\|_{\infty}, \| \xi_(\cdot; \mu))^{(1)} \|_{\infty}) \leq B_1 $ and $ \min_{z \in \calS_{\eps}} \xi  (z; \mu) \geq B_2 $.
In addition, define $ a(z; \mu) := \xi  (z; \mu) F(z; \mu) \in \mathbb{C}^m$ and we assume $a^{(m)}(z; \mu)$ is $\ell_{\xi}$-Lipschitz on $\calS_{\eps}$ for all $\mu \in \mathcal{N}_0$.
\end{assumption}

\ignore{
\begin{remark}
    The neighborhood size for $\mu$ in the above assumption can be relaxed as $\mathcal{N}_0 := \{\mu: \abs{\mu  - \mathring{\mu} }_{\infty} \leq (  T^{-\frac{ 2m +1} {(4m-1)(d +4)}}\land T^{- \frac{1} {d+8}} ) \sqrt{d \log T} \}$ if TDNN estimator is applied in the algorithm. 
\end{remark}
}

\begin{assumption}[Condition for $K(x)$] \label{assump:K}
The kernel function $K(x)$ and $K^{(1)}(x)$ are both $\ell_{K}$-Lipschitz continuous with bounded support. In addition, $K(x)$ is an order-m kernel, that is, $ \int K(x) \, \D x = 1 $, $\int x^j K(x) \, \D x = 0 $ for $j \in [m-1]$, and $\int | x^m K(x)| \, \D x  < \infty$. 
\end{assumption}

\begin{assumption}[Regularity condition] \label{assump:regu-phi}
There exists positive constants $L_{\phi}$ and $B_u$ such that $\phi^{(1)} (z) \ge L_{\phi}$ for all $  z \in \calS_{\eps} $ with $\phi(z)$ defined in \eqref{eqn:func-h}.
\end{assumption}

Assumptions \ref{assump:F}, \ref{assump:xi}, and \ref{assump:K} are standard assumptions for the smoothness of the relevant functions in nonparametric estimation and inference. The rate of convergence $(  T^{-\frac{ 2m +1} {(4m-1)(d +2)}}\land T^{- \frac{1} {2(d+4)}} ) \sqrt{d \log T}$ in Assumption \ref{assump:xi} can be satisfied for the DNN or TDNN estimators (with a further improved neighborhood size) with appropriately chosen parameters. 
Assumption \ref{assump:regu-phi} guarantees that $\phi(x)$ is strictly increasing and, thus, the uniqueness of the optimal solution of \eqref{eqn:opt-p}.
It is weaker than that in \cite{javanmard2019dynamic} since $1 - F(u)$ is log-concave under the further restriction of $L_{\phi'}\ge 1$.

The regret analysis unfolds in two steps. Firstly, we will develop the uniform convergence rate of both the DNN and TDNN estimator in estimating the mean utility function $\mathring{\mu}$, as well as the convergence rate of the Nadaraya-Watson estimator for estimating the noise distribution $F(\cdot)$. Secondly, building upon the uniform convergence rate of the estimators, we will obtain an upper bound for the regret of Algorithm \ref{algo:non-par-pricing}. 

\subsection{Statistical Convergence Rate of Learning}

We will establish statistical uniform convergence of non-parametric estimators for both the mean utility function and the market noise distribution. 
The finite-sample uniform bound for the DNN and TDNN estimators of $\mathring{\mu}(\bx)$ are presented below.

\begin{theorem} \label{le-DNN-TDNN}
Suppose density function $f (\cdot)$ of $\bx_t$ and the mean utility function $\mathring{\mu}(\cdot)$  are two times continuously differentiable with bounded second-order derivatives on the compact support $\mathcal{X}$ of $\bx_t$. The density $f (\bx)$ is bounded away from 0 and $\infty$ for $\bx \in \mathcal{X}$. Then we have for the DNN estimator that with probability $1 - 2\delta$, 
\begin{equation} \label{re-DNN}
\begin{aligned}
 \sup_{x \in \mathcal{X}} | \hat\mu_k^{\DNN}(\bx; s) - \mathring{\mu}(\bx) |  \leq  C  s^{-2/d} + B \sqrt{ \frac{2 s [\log \delta^{-1}  + \log d + d \log n ] } {n}} .
\end{aligned}
\end{equation} 
Furthermore, assume density function $f (\cdot)$ of $\bx_t$ and the mean utility function $\mathring{\mu}(\cdot)$  are four times continuously differentiable with bounded second, third, and fourth-order derivatives on the compact support $\mathcal{X}$ of $\bx_t$. Then we have for TDNN estimator  $D_n(s_1, s_2) (\bx) $ with $c_1   s_2 \leq s_1 \leq c_2 s_2$ that with probability $1 - \delta$, 
\begin{equation} \label{re-TDNN}
\begin{aligned}
 &\sup_{\bx \in \mathcal{X}} | \hat\mu_k^{\TDNN}(\bx; s_1, s_2) - \mathring{\mu}(\bx) | \\
 &\leq C s_1^{- \min\{3,\, 4/d\} }
 + C B \sqrt{ \frac{2 s_1 [\log \delta^{-1}  + \log d + d \log n ] } {n}} .
 \end{aligned}
\end{equation}
Here $C$ is a constant depending on the underlying density function $f(\cdot)$ of $\bx_t$, the mean utility function $\mathring{\mu} (\cdot)$, the ratio of $s_1/s_2$,  and the dimensionality d, which may take difference values from line to line. 
\end{theorem}

The proof of Theorem \ref{le-DNN-TDNN} can be found in Section \ref{append:proof-dnn-tdnn} of the supplementary material. The uniform convergence rate results in theorem \ref{le-DNN-TDNN} and the proofs are new to the literature and hold independent interest. Regarding the the uniform convergence of DNN estimator, the first error term $s^{-2/d}$ provides an upper bound for the bias, while the second error term corresponds to bound for the variance. In a similar manner, in the uniform bound for TDNN estimator, the first error term $s_2^{-4/d}$ is related to the bias and the second error term is associated with the variance. We can observe from Theorem \ref{le-DNN-TDNN} that the TDNN estimator preserves the same order of variance as the DNN estimator but reduces the bias. This improvement is achieved by TDNN's formulation of a linear combination of two DNN estimators, targeted at eliminating the first-order bias. For any general integer $l \geq 2$, one can employ an $l$-scale DNN estimator, which is defined as a linear combination of $l$ DNN estimators, each with a distinct subsampling scale to remove the bias up to the $(l-1)$-th order. This construction requires the mean utility function and density function of $\bx_t$ to have bounded derivatives up to the $2l$-th order. 

Building upon the uniform convergence rate for DNN estimators, we obtain the following result on the accuracy of the kernel estimator $\hat{F}_k\paren{z; \hat{\mu}_k^{\DNN}} $ defined in \eqref{eqn:F-est} for estimating the noise distribution $F(z)$. 

\begin{theorem} \label{lemma-F-DNN}
Assume DNN estimator $\hat{\mu}_k^{\DNN}$ is applied with subsampling size $s = O(n_{k, \exp}^{\frac{d} {d + 4}})$.
Under condition of Theorem \ref{le-DNN-TDNN} for DNN estimator and Assumptions \ref{assump:F}--\ref{assump:regu-phi}, there exist constants $C_1> 0$ and $C_2 > 0$ such that for any episode $k$ satisfying $n_{k, exp}^{\frac{2m} {2m + 1}} \geq \max(C_1 d ( \frac{\log n_{k, exp}} {2m + 1} + 1 ), 3 \log n_{k, exp})$ with probability at least $1 - 2\delta$ with $\delta \in (4 \exp \{- C_2 n_{k, exp}^{\frac{m} {2m + 1}} / \log n_{k, exp} \}, 1/2]$, it holds that 
\begin{equation} \label{eqn:kernel-accuracy}
\begin{aligned}
\sup_{z \in \calS_{\eps}, {\mu}\in \mathcal{N}_{k}} \abs{ \hat{F}_k\paren{z; \hat{\mu}_k^{\DNN}}  - F(z) } 
&\leq  B_F n_{k, exp}^{- \frac{m} {2m + 1}} \sqrt{\log n_{k, exp} } ( \sqrt d + \sqrt{\log \delta^{-1} } ) \\
&+ B_F {n_{k, exp}^{-\frac{2}{ d + 4}} \sqrt {d \log n_{k, exp}}}.
\end{aligned}
\end{equation}
\end{theorem}

Theorem \ref{lemma-F-DNN} follows directly from applying Lemma \ref{lemma-F} in Appendix \ref{append:proof-regret} and Theorem \ref{le-DNN-TDNN}. 
The first term in the upper bound comes from kernel approximation of the noise distribution, while the second term comes from the upstream estimation error of DNN estimator for mean utility $\mathring{\mu}(\bx)$. 
For existing work that assumes a parametric model of $\mathring{\mu}(\bx)$, the second error term is dominated by the first term (\cite{fan2022policy}). 
However, Theorem \ref{lemma-F-DNN} shows that the upstream estimation error of $\mathring{\mu}(\bx)$ becomes non-negligible as the dimensionality $d$ grows under the more general non-parametric setting. 

When TDNN estimator is applied to estimate the nonparametric mean regression function $\mathring{\mu}(\cdot)$ and then the same Nadaraya-Watson kernel estimator for the error distribution $F(\cdot)$, we can obtain the following parallel result in Theorem \ref{lemma-F-TDNN} below for the kernel estimator $\hat{F}_k\paren{z; \hat{\mu}_k^{\TDNN}} $ when TDNN estimator is applied to estimate the mean utility function. 

\begin{theorem} \label{lemma-F-TDNN}
Assume TDNN estimator $\hat{\mu}_k^{\TDNN}$ is applied with subsampling scales $s_1 = O(n_{k, \exp}^{\frac{d} {d + 8} \lor \frac 1 7 })$ and $s_2 \asymp s_1$. Under   condition of Theorem \ref{le-DNN-TDNN} for TDNN estimator and Assumptions \ref{assump:F}--\ref{assump:regu-phi}, there exist constants $C_1> 0$ and $C_2 > 0$ such that for any episode $k$ satisfying $n_{k, exp}^{\frac{2m} {2m + 1}} \geq \max(C_1 d ( \frac{\log n_{k, exp}} {2m + 1} + 1 ), 3 \log n_{k, exp})$ and $\delta \in (4 \exp \{- C_2 n_{k, exp}^{\frac{m} {2m + 1}} / \log n_{k, exp} \}, \frac{1}{2}]$, it holds with probability at least $1 - 2\delta$ that 
\begin{equation} \label{eqn:kernel-accuracy}
\begin{aligned}
\sup_{z \in \calS_{\eps}, {\mu}\in \mathcal{N}_{k}} \abs{ \hat{F}_k\paren{z; \hat{\mu}_k^{\TDNN}}  - F(z) } 
&\leq  B_F n_{k, exp}^{- \frac{m} {2m + 1}} \sqrt{\log n_{k, exp} } ( \sqrt d + \sqrt{\log \delta^{-1} } ) \\
&+ C B_F {n_{k, exp}^{- (\frac{4} {d + 8} \land \frac{3}{7}) } \sqrt {d \log n_{k, exp}}}.
\end{aligned}
\end{equation}
\end{theorem}

\subsection{Regret Analysis of Decision Making} \label{sec:regret}
\noindent The following result establishes an upper bound for the cumulative regret of Algorithm \ref{algo:non-par-pricing} over a time horizon $T$ when DNN estimator is applied in the mean utility function estimation. The proofs of Theorem \ref{thm-1} and \ref{thm-TDNN} are provided in Appendix \ref{append:proof-regret}, respectively.  

\begin{theorem} \label{thm-1}
Assume the DNN estimator is applied with subsampling size $s = O(n_{k, exp}^{\frac{d} {d + 4}})$ to estimate the mean utility function.
Suppose condition of Theorem \ref{le-DNN-TDNN} for DNN estimator and Assumptions \ref{assump:F}--\ref{assump:regu-phi} hold. Let the exploration phase size in the $k$th episode be $n_{k, exp} = n_k^{\frac {2m +1} {4m-1} \lor \frac{d + 4}{d + 8} }$ in Algorithm \ref{algo:non-par-pricing}, then for any $T  \geq [\max(C_1 d (1 + \frac{\log T} {2m+1}),  3 \log T)]^{ \frac{2m+ 1}{m} ( \frac{4m-1}{2m +1} \land \frac{d + 8}{d + 4} ) }$, the total regret satisfies 
\begin{equation} \label{result-thm-DNN-regret}
\Reg(\pi, T)   \leq  C T^{ \frac{2m+1}{4m-1} \lor \frac{d + 4}{d+8} }    [(\log T)^2  + d \log T] .
\end{equation}
\end{theorem}

\begin{remark}
Theorem \ref{thm-1} requires that the mean utility function $\mathring{\mu}(\cdot)$ is twice differentiable and the error distribution $F(\cdot)$ has $m$-th order derivatives. 
The regret upper bound results from a synergistic learning of $\mathring{\mu}(\cdot)$ and $F(\cdot)$. 
Specifically, learning a twice-differentiable $\mathring{\mu}(\cdot)$ incurs a regret bound of $\tilde{\cO}(T^{\frac{2m+1}{4 m -1}})$, while learning an $m$-th differentiable $F(\cdot)$ incurs a regret bound of $\tilde\cO(T^{\frac{d + 4}{d + 8}})$. 
Overall, the regret bound is dominated by $\tilde{\cO}(T^{\frac{2m+1}{4 m -1}})$ as $1 \leq d \leq \frac{6}{m-1}$ and is instead dominated by $\tilde{\cO}(T^{\frac{d + 4}{d + 8}})$ as $d > \frac{6}{m-1}$.

If the mean utility function $\mathring{\mu}(\cdot)$ is further fourth differentiable, the more sophisticated TDNN estimator can be utilized to estimate $\mathring{\mu}(\cdot)$ in Algorithm \ref{algo:non-par-pricing}, achieving a lower regret bound due to its debiasing capability. The corresponding regret bound of incorporating the TDNN estimator are detailed in Theorem \ref{thm-TDNN} below.

\end{remark}

\begin{theorem} \label{thm-TDNN}
Assume the TDNN estimator is applied with subsampling sizes $s_1 = O(n_{k, exp}^{\frac{d} {d + 8} \lor \frac{1}{7}})$ and $s_2 \asymp s_1 $ for mean utility function estimation. Suppose condition of Theorem \ref{le-DNN-TDNN} for TDNN estimator and Assumptions \ref{assump:F}--\ref{assump:regu-phi} hold. Let the exploration phase size in the $k$th episode be $n_{k, exp} = n_k^{\frac {2m +1} {4m-1} \lor \frac{d + 8}{d + 16} \lor \frac{7}{13} }$ in Algorithm \ref{algo:non-par-pricing}, then for any $T  \geq [\max(C_1 d (1 + \frac{\log T} {2m+1}),  3 \log T)]^{ \frac{2m+ 1}{m} ( \frac{4m-1}{2m +1} \land \frac{d + 16}{d + 8} \land \frac{13}{7} ) }$, the total regret satisfies 
\begin{equation} \label{result-thm-TDNN-regret}
\Reg(\pi, T)   \leq  C T^{\frac{2m+1}{4m-1} \lor \frac{d + 8}{d + 16} \lor \frac{7}{13}}    [(\log T)^2  + d \log T].
\end{equation}


\end{theorem}

\begin{remark}
The overall regret of doubly-nonparametric learning and pricing is a trade-off between two components: the term of order $\tilde{\cO} ( T^{\frac{d + 8}{d + 16} \lor \frac{7}{13}} )$ corresponds to learning a fourth differentiable mean utility function $\mathring{\mu}(\bx)$, while the term of order $\tilde{\cO} ( T^{\frac{2m + 1}{4 m - 1}})$ corresponds to learning an $m$-th differentiable distribution function $F(\cdot)$. 

When $F(\cdot)$ is comparatively more difficult to learn, i.e.~less smooth in the sense that $m  <  1 +  \min(\frac{12} {d}, 9)$, the overall regret is dominated by $\tilde{\cO} ( T^{\frac{2m + 1}{4 m - 1}})$. 
When $\mathring{\mu}(\bx)$ is comparatively more difficult to learn, i.e.~higher-dimensional such that $m  \geq 1 +  \min(\frac{12} {d}, 9)$, the the overall regret is dominated by $\tilde{\cO} ( T^{\frac{d + 8}{d + 16} \lor \frac{7}{13}})$ where the term $\frac{7}{13}$ corresponds to the setting of $d=1$. 

\end{remark}

\begin{remark}\label{rmk:why lower regret}

The structure of the random utility model leads to \eqref{eqn:opt-p} and \eqref{eqn:func-h}, facilitating an advanced analysis of regret, which demonstrates a high-order dependence on the square of the difference between the optimal and estimated prices. 
To elaborate, consider the revenue function $\rev_t(p)$ defined in \eqref{eqn:rev_t} along with its first and second derivatives, $\rev_t^{(1)}(\cdot)$ and $\rev_t^{(2)}(\cdot)$ respectively. 
As $p_t^*$ is the price that maximizes revenue, it holds that $\rev_t^{(1)}(p_t^*) = 0$. A Taylor expansion around $p_t^*$ yields
\begin{equation*} 
\rev_t(p_t) = \rev_t(p_t^*) + \frac{1}{2}\rev_t^{(2)}(p)(p_t-p_t^*)^2, 
\end{equation*}
for some $p$ between $p_t$ and $p_t^*$. 
Additionally, according to \eqref{eqn:func-h}, the $(p_t-p_t^*)^2$ is directly linked to the square of the estimation errors in $F$ and $\mathring{\mu}$.

This nuanced relationship results in improved regret bounds under various scenarios, surpassing those reported in \cite{chen2021nonparametric} and \cite{bu2022context}. Unlike the linear dependence on the difference between optimal and estimated prices observed in the discretization-based methods used in these studies, our model's regret depends quadratically on these differences, leading to potentially tighter regret bounds under comparable conditions.
\end{remark}

\begin{remark}
Theorem \ref{thm-1} and \ref{thm-TDNN} require only second- and fourth-order smoothness, respectively, of the true utility function $\mathring{\mu}(\cdot)$. When the smoothness of $\mathring{\mu}(\cdot)$ reaches a general level $\ell$ for an even integer $\ell$, it is possible to develop a $(\ell/2)$-scale DNN estimator analogous to the TDNN framework that achieves better convergence rates. 
\end{remark}

\begin{remark}
When we consider the policy set $ \Pi := \{\pi: \pi(p_t) =  \hat{\mu}_k (\bx_t) + \hat{\phi}_k^{-1} (- \hat{\mu}_k (\bx_t) ) \} $ with $\hat{\phi}_k(z) = z - \frac{1 - \hat{F}_k(z)}{ \hat{F}_k^{(1)} (z)}$, where $\hat{\mu}_k$, $\hat{F}_k$, and $\hat{F}^{(1)}_k$ are estimators for the mean utility function, noise distribution, and derivative of the noise distribution in the $k$-th episode, respectively, then the optimality of the regret upper bound reduces to the optimality of these estimators. Based on \cite{Stone1982}, our nonparametric estimators DNN or TDNN $\hat\mu_k$ , Nadaraya-Watson kernel estimator $\hat{F}$ and its derivative $\hat{F}^{(1)}$, all achieve optimal rate of convergence. Therefore, the proposed algorithm using those estimators is minimax optimal when we only focus on the policy set $\Pi$. 
\end{remark}

\begin{remark}
 Our analysis can be generally extended to a broader range of nonparametric estimators for the mean utility function, such as kernel estimators, random forest, and neural network methods, provided that these estimators achieve a similar uniform convergence rate as presented in Theorem \ref{le-DNN-TDNN} for the DNN and TDNN estimators. 
\end{remark}


\section{Empirical Analysis} 
\subsection{Synthetic Data}
\label{sec:simul}
\noindent In this subsection, we conduct large-scale simulations to illustrate the efficiency of our policy. Recall that the pricing model is characterized by \eqref{eqn:nonpar-utility} and \eqref{eqn:y-logistic}. To elaborate, the mean utility function is set to be $\mathring{\mu}(\bx_t) = 2(x_{t1}-1)^2 + 2(x_{t2}-1)^2 + 2(x_{t3}-1)^2 $, where the contextual covariates $\{\bx_t\}_{t \geq 1}$ are i.i.d. bounded random vectors of dimension $d=3$, and each vector consists of independent entries uniformly distributed according to $\Unif(0, 2)$.
The noise $\{\eps_t\}$ are i.i.d. with probability density $f(z) \propto (1/4-z^2) \cdot\bbone(\abs{z}\le1/2),$ where the smoothness order $m=2$.

When implementing algorithm \ref{algo:non-par-pricing}, we divide the time horizon $T$ into consecutive episodes by setting the length of the $k$-th episode as $n_k=2^{k-1}n_0$ with $k\in\NN^{+}$ and $n_0=200$. We further separate every episode into an exploration phase with length $|T_{k,exp}|=\lfloor(dn_k)^{\frac{2m + 1}{4m - 1}}\rfloor$ induced from Theorem \ref{thm-1}, then the exploitation phase contains the rest of the time in that episode. In the exploration phase, we sample $p_t$ from $\textrm{Unif}(0,B=6.5)$, since $B=6.5$ is a valid upper bound of $v_t$. In the exploitation phase, we set the kernels used in \ref{eqn:aux-def-a-xi} as follows:  $K(z)=(1-11/3z^2)(1 - z^2)^3 \cdot \bbone(\abs{z}\le1/2)$. In episode $k$, we set the bandwidth $b_k$ as $6\cdot |T_{k,exp}|^{-\frac{1}{2m+1}}$ with $m=2$ according to the assumptions in the theoretical analysis. In reality, one can also tune the bandwidth and the order of kernel $K(\cdot)$ by using cross-validation at the end of every exploration phase. Moreover, when calculating $p_t=\hat{h}_k \circ \hat{\mu}_k(\bx_t) =\hat{\mu}_k(\bx_t)+\hat\phi_k^{-1}(-\hat{\mu}_k(\bx_t))$, we find $\hat\phi_k^{-1}(-\hat{\mu}_k(\bx_t))$ as follows: First, we look for $x\in[-1,1]$ such that $\hat\phi_k(x)=-\hat{\mu}_k(\bx_t)$ (We choose a large interval $[-1,1]$ that contains the true support of $\phi(x)$ [-0.5, 0.5], since in reality, we might only know a wider range of the true support). Then, we make a transformation of variable $x$ to $x(y)=-2\cdot \exp(y)/(1+\exp(y))+1$ and solve $y$ as the root of $\hat\phi_k(x(y))+\hat{\mu}_k(\bx_t)=0$ by using Newton's method starting at $y=0$. Finally, we set $x=-2\cdot \exp(y)/(1+\exp(y))+1$ as $\hat\phi_k^{-1}(-\hat{\mu}_k(\bx_t))$ and offer $p_t$ according to the algorithm.

In the simulation, we also compare the efficiency of our methods with two closely related approaches: the 'RMLP-2' proposed by \citep{javanmard2019dynamic}, and the linear model based policy introduced by \citep{fan2022policy}. They both address the same problem but assume the underlying mean utility function to be the linear model and, hence, fit a linear regression during the exploration phase. Besides, the 'RMLP-2' additionally assumes the noise distribution belongs to a parametric function class, while given the algorithm's lack of access to the true noise distribution, we adopt a Gaussian distribution assumption for the noise. 
We conduct repeated trials of Algorithm \ref{algo:non-par-pricing}, the linear model based policy and the 'RMLP-2', respectively, from $n_1=200$ to $n_8=12800$, to document the cumulative regret $\Reg(\pi,T)=\sum_{t = 1}^T\Reg(\pi,t)$ with $T=1$ to  $\sum_{k = 1}^8 2^{k-1}*200$. Consequently, the averaged regret is stored as $\text{AveReg}(T)= \Reg(\pi,T)/T.$ And the illustration of the distinctions in performance is detailed in several figures. Specifically, we plot $\Reg(\pi,T)$ against $T$ in Figure \ref{fig:sim cumregret}, and the averaged regret $\text{AveReg}(T)$ against $T$ in Figure \ref{fig:sim averaged regret}. 
Recall from Theorem \ref{thm-1} that the cummulative regret $\Reg(\pi, T)\lesssim [(\log T)^2  + d \log T ]T^{\frac{2m + 1}{4m -1}}$ when $d \le 6/(m-1)$. Hence, we plot $\tilde{\textrm{reg}}(T)$ against $\log(T)-\log(1500)$ in Figure \ref{fig:sim logregret}, where $\tilde{\textrm{reg}}(T):=\log(\Reg(T)) - (2\log\log T+\log3\log T) - (\log(\Reg(T=1500)) - 2\log\log 1500 - \log3\log 1500)$. 
In addition, we plot the averaged regret of the exploitation phase against $T_{com}$ in Figure \ref{fig:sim relative regret}, where the y-axis Averaged Regret (Exploitation Phase) is defined as: $\text{AveReg}(\calT_{com}) =\sum\limits_{t \in \calT_{com}} Reg(t)/T_{com}$, and $T_{com}=|\calT_{com}|.$

\begin{figure}[H]
\centering
\subfloat[Cumulative Regret against T.]{
  \includegraphics[width=0.45\textwidth]{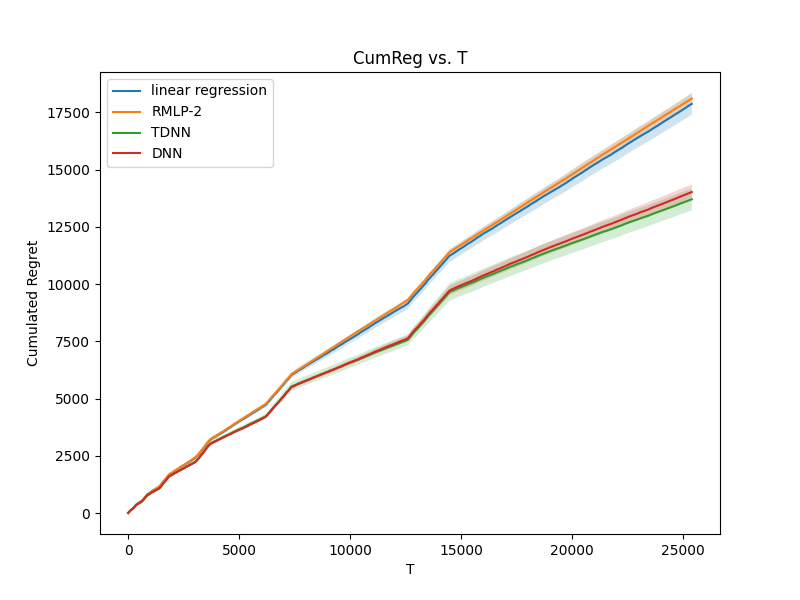}
  \label{fig:sim cumregret}
}
\subfloat[Regret log-log plot.]{
  \includegraphics[width=0.45\textwidth]{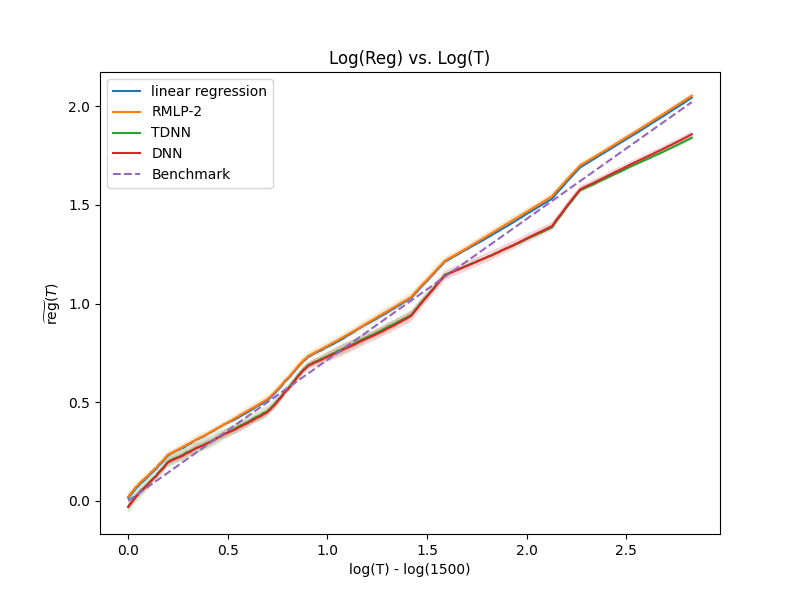}
  \label{fig:sim logregret}
}

\medskip

\caption{(a): The x-axis is T while the y-axis is $\Reg(\pi, T)$. (b): The x-axis is 
\( \log(T) - \log(1500) \), while the y-axis 
is \( \Tilde{\text{reg}}(T)\) defined above. For (a) and (b), the solid lines represent the mean $\Reg(\pi, T)$ and $\Tilde{\text{reg}}(T)$, respectively, over 20 independent runs. The light color areas 
around those solid lines depict the standard error of the estimations. The blue, green, and red lines represent linear regression, TDNN, and DNN estimations in the exploration phase in Algorithm 1, respectively. And the orange line represents applying 'RMLP-2'. 
The dashed lines represent the benchmark whose slopes are equal to \( \frac{2m+1}{4m-1} \) with \( m=2 \).}
\label{fig:1}
\end{figure}

\begin{figure}[H]
\centering
\subfloat[$\text{AveReg}(T)$ against T.]{
  \includegraphics[width=0.45\textwidth]{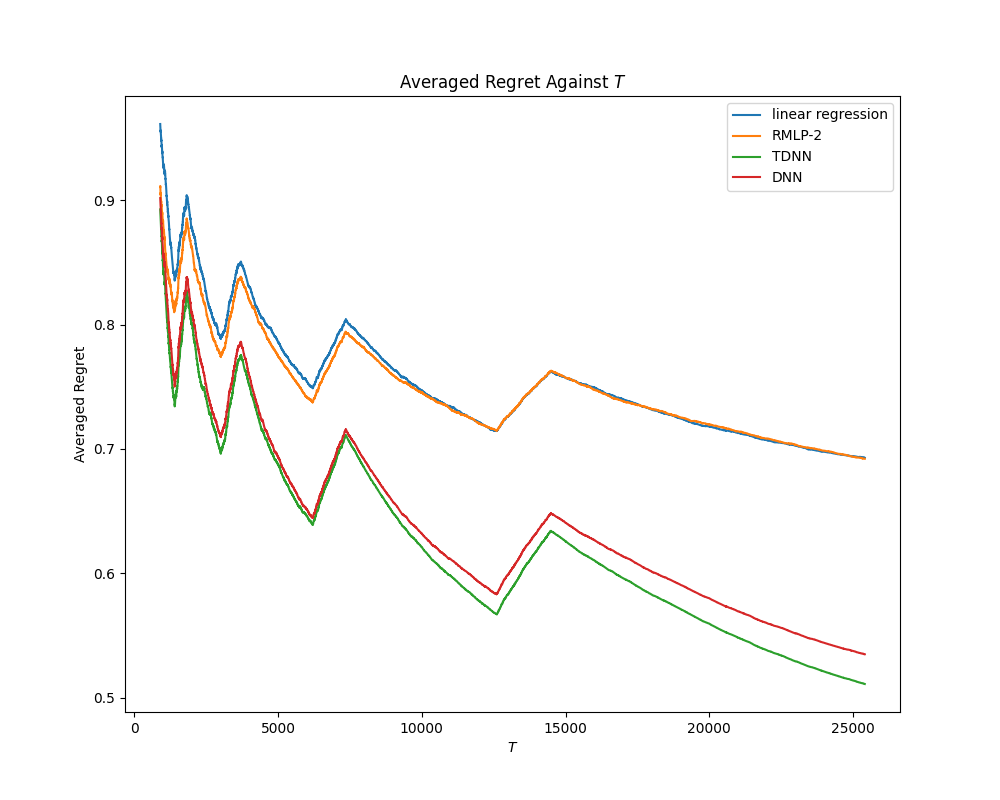}
  \label{fig:sim averaged regret}
}
\subfloat[$\text{AveReg}(T_{com})$ against $T_{com}$.]{
  \includegraphics[width=0.45\textwidth]{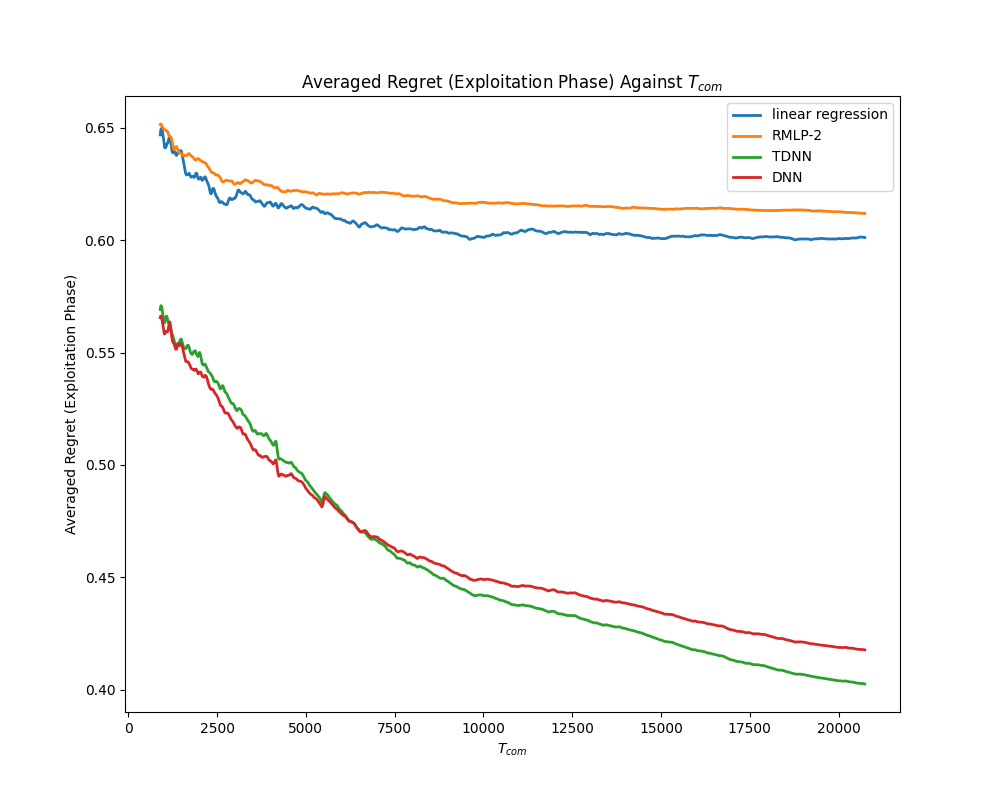}
  \label{fig:sim relative regret}
}

\medskip

\caption{For (a) and (b), the solid lines represent the mean $\text{AveReg}(T)$ and $\text{AveReg}(\calT_{com})$ respectively, over 20 independent runs. The remaining caption is the same as Figure \ref{fig:1}.}
\label{fig:2}
\end{figure}

From Figure \ref{fig:1}, we conclude that under the specified setting, the rates of the empirical cumulative regrets produced by Algorithm \ref{algo:non-par-pricing} (as shown by the solid green and red lines) do not exceed their theoretical counterparts given in Theorems \ref{thm-1}(as shown by the dashed lines). In many cases, the growth rates of the empirical regrets are very close to those of the theoretical line. This demonstrates the tightness of our theoretical results. We also see that the regrets we achieved are much smaller than those two approaches. As for the comparison with linear regression, TDNN and DNN estimations are robust to the mis-specification of the mean utility function, and in comparison with 'RMLP-2', our method is robust to the mis-specification of the noise distribution since our algorithm can adapt to all noise distributions in the non-parametric class. From Figure \ref{fig:2}, we conclude that under the same setting, Algorithm \ref{algo:non-par-pricing} incorporating TDNN estimation surpasses all other methods in performance, since as the episode length increases, the TDNN method accelerates the convergence of regret during the exploitation phase with more accurate estimation of the underline mean utility function and consequently, more accurately approaching towards noise distribution.  In contrast, the decay in the averaged regret of exploitation phase for the two benchmark approaches nearly goes to zero. This occurs because the increasing size of sample points during the exploration phase does not enhance the estimation accuracy of the mean utility function due to mis-specification.

\subsection{Real Data Analysis}
\label{sec:real}

\noindent We utilize the real-life auto loan dataset from Columbia University's Center for Pricing and Revenue Management, a dataset that has also been employed in several prior studies \citep{Phillips2015Effectiveness,ban2021personalized,luo2024distribution,Wang2023OnlineRegularization}. This dataset encompasses 208,085 auto loan applications spanning from July 2002 to November 2004 and includes various features, such as loan amount and borrower information. For consistency, we perform feature selection in line with existing work \citep{ban2021personalized,luo2024distribution,Wang2023OnlineRegularization}, focusing on the following four features: approved loan amount, FICO score, prime rate, and competitor's rate. Regarding the price variable, we also computed it in the same way as the aforementioned literature, where $p_t=\textrm{Monthly Payment}\cdot \sum_{t=1}^{\textrm{Term}}(1+\textrm{Rate})^{-t}-\textrm{Loan Amount}$. The rate is set as $0.12\% $, which is an approximate average of the monthly London interbank rate for the studied time period.
Moreover, this dataset also records the purchasing decisions of the borrowers, given the price set by the lender. 

Note that one is not able to obtain online responses to any algorithms. Thus, we follow the calibration idea proposed in \citep{ban2021personalized,luo2024distribution,Wang2023OnlineRegularization} to first learn a binary choice model using the entire dataset and then leverage it as the ground truth to conduct numerical experiments. To make calibration model more stable, we process all covariates by standardization and take log of the price variable. One can also remove the extreme values that are outside the 0.01 and 0.99 percentiles to mitigate the influence of outliers. For learning the decision model, we typically fit a generalized logistic model with a quadratic form as the logits; meanwhile, instead of using a conventional logistic function, we adopt a specified sigmoid function with bounded support. Specifically, denote the contextual covariate as $\bx_t \in \RR^4, \; t=1,2,\dots$ and the quadratic form $\mathring{\mu}(\bx_t) = \bbeta^\top \bx_t^2, \; \bbeta \in \RR^4$ for the mean utility function. Then the noise $\{\eps_t\}$ are set to independently follow the distribution with probability density function $f(z) \propto (1-(\frac{z}{3})^2)^4 \cdot\bbone(\abs{z}\le3).$ Furthermore, with the belief that the coefficients of the mean utility function can be varied across the data while the noise shares the same distribution. We employ a two-stage approach learning calibration model. Initially, the covariates $\{\bx_t\}_{t\ge 1}$ are divided into 20 clusters based on similarities, utilizing the K-means clustering technique. Here we denote $\bx_t^{(k)}$ if the covariate $\bx_t$ belongs to the $k$-th cluster. Then, we apply generalized logistic regression to each data cluster, fitting the coefficients of the quadratic form for each group while maintaining the same noise distribution. Figure
\ref{fig: noise distribution} demonstrates the alignment between observed data and the theoretical model over the noise $\eps_t$, which means the calibration model is well specified.

\begin{figure}[H]
\centering
\subfloat[Empirical vs. Theoretical(calibration model) probability distributions.]{
  \includegraphics[width=0.45\textwidth]{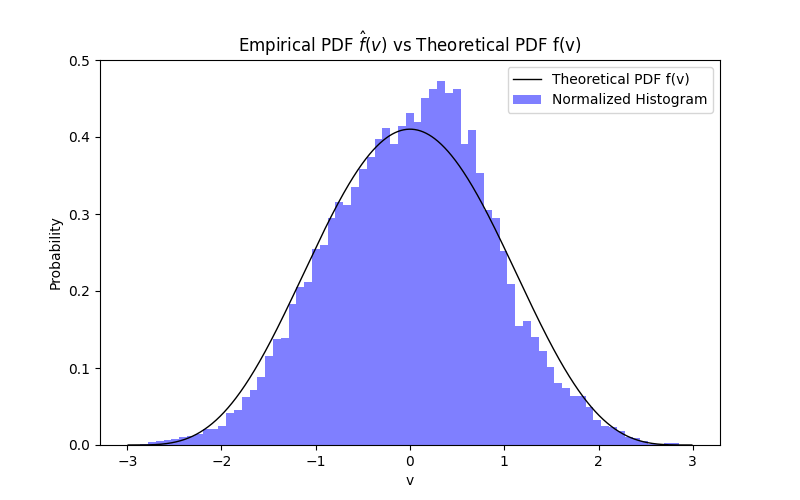}
  \label{fig:empirical_pdf}
}
\subfloat[Empirical vs. Theoretical (calibration model) cumulative distributions.]{
  \includegraphics[width=0.45\textwidth]{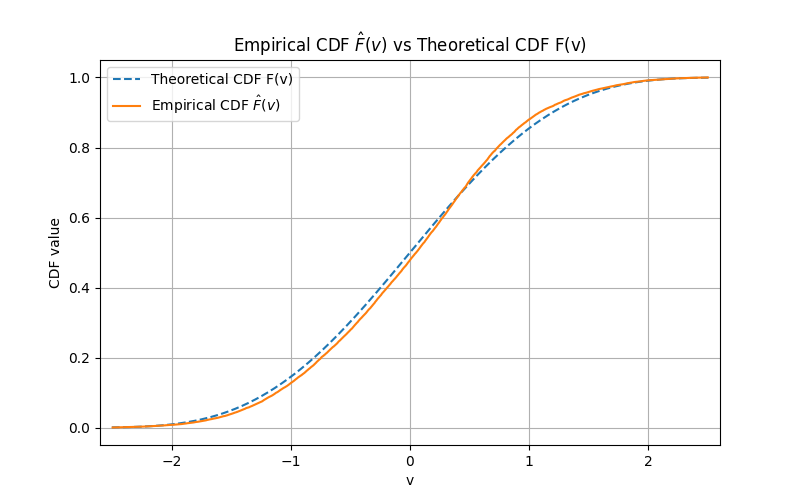}
  \label{fig:empirical_cdf}
}

\medskip

\caption{Comparison of Empirical and Theoretical Probability Distributions.}
\label{fig: noise distribution}
\end{figure}

Given these key components, the remaining experiments remain almost the same as discussed in the simulation part. Here we use the calibration model, setting $\mathring{\mu}(\bx_t)$ with parameter $\bbeta_k, k=1, \dots, 20$, and the noise distribution $f(v)$ given above and sample $\bx_t$ from those four features. Since in reality we don't know the true smoothness order of noise distribution, in each episode, we set the exploration phase with length $|\calT_{k,exp}| =\lfloor(dn_k)^{\frac{d+4}{d+8}}\rfloor$. We set $B=6, n_0=200$ and conduct Algorithm \ref{algo:non-par-pricing}. Recall from Theorem \ref{thm-1} that the regret $\Reg(\pi, T)\lesssim [(\log T)^2  + d \log T ]T^{\frac{d + 4}{d + 8}}.$ Hence, we plot $\tilde{\textrm{reg}}(T)$ against $\log(T)-\log(1500)$ in Figure \ref{fig:real logregret}, where $\tilde{\textrm{reg}}(T):=\log(\Reg(T)) - (2\log\log T+\log4\log T) - (\log(\Reg(T=1500)) - 2\log\log 1500 - \log4\log 1500)$. 

\begin{figure}[H]
\centering
\subfloat[Cumulative Regret against T.]{
  \includegraphics[width=0.45\textwidth]{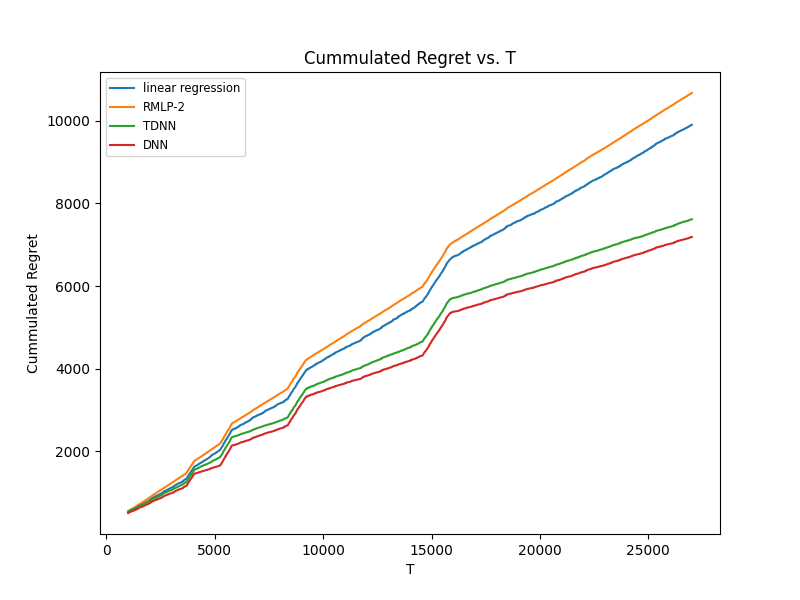}
  \label{fig:real cumregret}
}
\subfloat[Regret log-log plot.]{
  \includegraphics[width=0.45\textwidth]{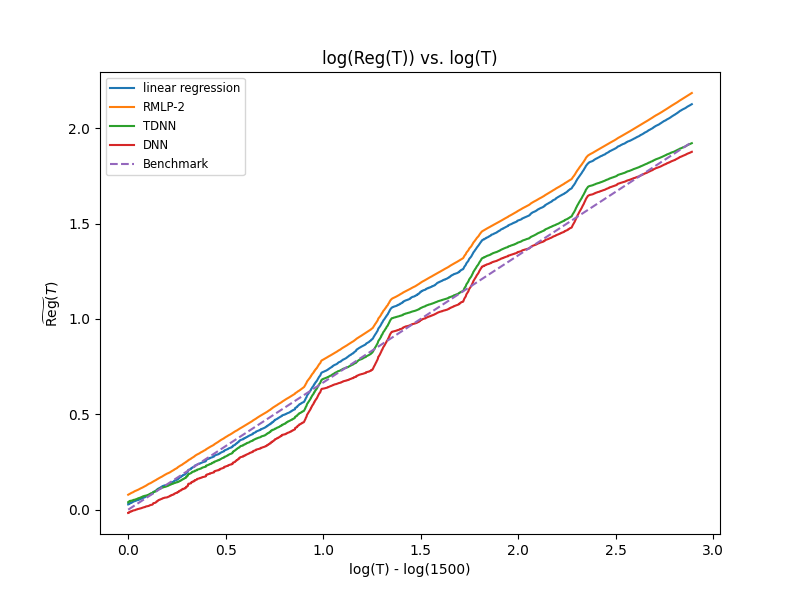}
  \label{fig:real logregret}
}

\medskip

\caption{Comparison between our methods and linear model based policy and ‘RMLP-2’ based on real data application. (a): The x-axis is T while the y-axis is $\Reg(\pi, T)$. (b): The x-axis is 
\( \log(T) - \log(1500) \), while the y-axis 
is \( \Tilde{\text{reg}}(T)\) defined above. The dashed lines has slope equal to \( \frac{d+4}{d+8} \) with \( d=4 \). Here we only run algorithm one time.}
\end{figure}

\begin{figure}[H]
\centering
\subfloat[Averaged Regret against T.]{
  \includegraphics[width=0.45\textwidth]{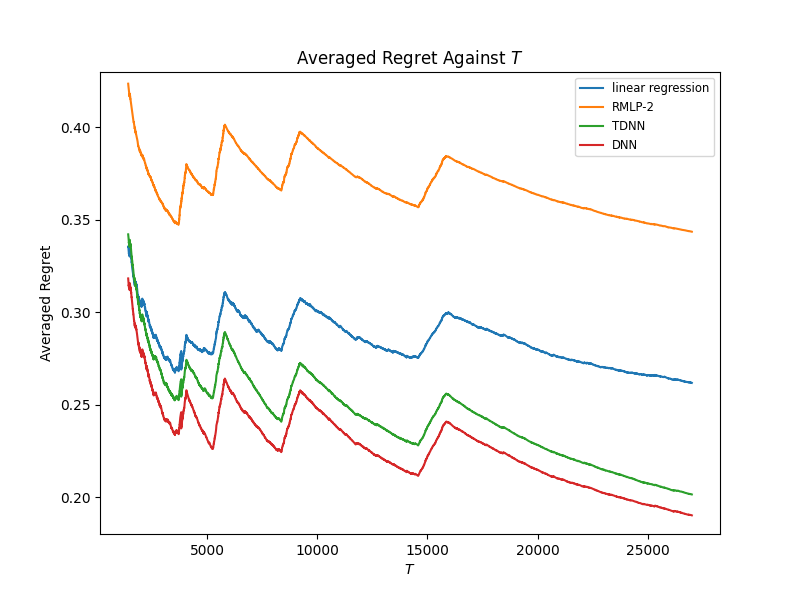}
  \label{fig:real averaged regret}
}
\subfloat[$\text{AveReg}(T_{com})$ against $T_{com}$.]{
  \includegraphics[width=0.45\textwidth]{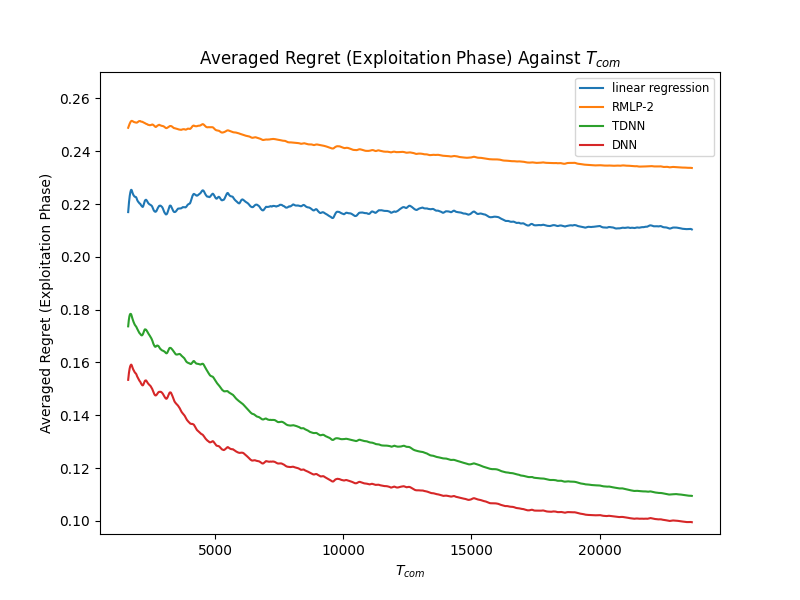}
  \label{fig:real relative regret}
}

\medskip

\caption{Comparison between our methods and linear model based policy and ‘RMLP-2’ based on real data application. For (a) and (b), the solid lines represent the $\text{AveReg}(T)$ and $\text{AveReg}(\calT_{com})$, respectively.}
\end{figure}

To summarize, our methods outperform the 'RMLP-2' \citep{javanmard2019dynamic}, and the linear model based policy \citep{fan2022policy} in terms of both the regret performance and the ability to adapt to diﬀerent mean utility models and noise distributions.

\section{Conclusion} \label{sec:conclusion}

\noindent This paper advances the field of dynamic pricing by introducing and exploring doubly nonparametric random utility models. By forgoing parametric assumptions for both the mean utility function and the noise distribution, our approach reduces the risk of model misspecification and increases the reliability of pricing strategies. Employing advanced nonparametric estimation techniques like Distributional Nearest Neighbors and Nadaraya-Watson kernel estimation, we enhance the accuracy of dynamic pricing decisions in complex market environments. Our theoretical contributions, which include new analytical results and insights into the interaction between model components, offer valuable guidance for both academic research and practical applications. This work not only deepens the understanding of nonparametric estimation in pricing models but also provides a robust framework for businesses to optimize pricing dynamically and effectively in data-driven markets.

%
%

\clearpage
\setstretch{1.68}
\bibliographystyle{\mybibsty}
\bibliography{\mybib}

\begin{thebibliography}{}

\bibitem[\protect\citeauthoryear{Ai, Ku{\v{z}}elka, and Wang}{Ai
  et~al.}{2022}]{ai2022hoeffding}
Ai, J., O.~Ku{\v{z}}elka, and Y.~Wang (2022).
\newblock Hoeffding and bernstein inequalities for u-statistics without
  replacement.
\newblock {\em Statistics \& Probability Letters\/}~{\em 187}, 109528.

\bibitem[\protect\citeauthoryear{Ban and Keskin}{Ban and
  Keskin}{2021}]{ban2021personalized}
Ban, G.-Y. and N.~B. Keskin (2021).
\newblock Personalized dynamic pricing with machine learning: High-dimensional
  features and heterogeneous elasticity.
\newblock {\em Management Science\/}~{\em 67\/}(9), 5549--5568.

\bibitem[\protect\citeauthoryear{Bastani, Simchi-Levi, and Zhu}{Bastani
  et~al.}{2021}]{bastani2021meta}
Bastani, H., D.~Simchi-Levi, and R.~Zhu (2021).
\newblock Meta dynamic pricing: Transfer learning across experiments.
\newblock {\em Management Science\/}~{\em 68\/}(3), 1865--1881.

\bibitem[\protect\citeauthoryear{Biau, C{\'e}rou, and Guyader}{Biau
  et~al.}{2010}]{biau2010rate}
Biau, G., F.~C{\'e}rou, and A.~Guyader (2010).
\newblock On the rate of convergence of the bagged nearest neighbor estimate.
\newblock {\em Journal of Machine Learning Research\/}~{\em 11\/}(2).

\bibitem[\protect\citeauthoryear{Bu, Simchi-Levi, and Wang}{Bu
  et~al.}{2022}]{bu2022context}
Bu, J., D.~Simchi-Levi, and C.~Wang (2022).
\newblock Context-based dynamic pricing with separable demand models.
\newblock {\em Available at SSRN\/}.

\bibitem[\protect\citeauthoryear{Chen, Chao, and Shi}{Chen
  et~al.}{2021}]{chen2021jointpricing}
Chen, B., X.~Chao, and C.~Shi (2021).
\newblock Nonparametric learning algorithms for joint pricing and inventory
  control with lost sales and censored demand.
\newblock {\em Mathematics of Operations Research\/}~{\em 46\/}(2), 726--756.

\bibitem[\protect\citeauthoryear{Chen and Gallego}{Chen and
  Gallego}{2021}]{chen2021nonparametric}
Chen, N. and G.~Gallego (2021).
\newblock Nonparametric pricing analytics with customer covariates.
\newblock {\em Operations Research\/}~{\em 69\/}(3), 974--984.

\bibitem[\protect\citeauthoryear{Chen and Gallego}{Chen and
  Gallego}{2022}]{chen2022primaldual}
Chen, N. and G.~Gallego (2022).
\newblock A primal–dual learning algorithm for personalized dynamic pricing
  with an inventory constraint.
\newblock {\em Mathematics of Operations Research\/}~{\em 47\/}(4), 2585--2613.

\bibitem[\protect\citeauthoryear{Chen and Hu}{Chen and Hu}{2023}]{chen2023data}
Chen, N. and M.~Hu (2023).
\newblock Data-driven revenue management: The interplay of data, model, and
  decisions.
\newblock {\em Service Science\/}.

\bibitem[\protect\citeauthoryear{Chen, Jasin, and Shi}{Chen
  et~al.}{2022}]{chen2022elements}
Chen, X., S.~Jasin, and C.~Shi (2022).
\newblock {\em The Elements of Joint Learning and Optimization in Operations
  Management}, Volume~18.
\newblock Springer Nature.

\bibitem[\protect\citeauthoryear{Chen, Simchi-Levi, and Wang}{Chen
  et~al.}{2021}]{Chen2021Privacy}
Chen, X., D.~Simchi-Levi, and Y.~Wang (2021).
\newblock Privacy-preserving dynamic personalized pricing with demand learning.
\newblock {\em Management Science\/}~{\em 68\/}(7), 4878--4898.

\bibitem[\protect\citeauthoryear{Cohen, Lobel, and Paes~Leme}{Cohen
  et~al.}{2020}]{cohen2020feature}
Cohen, M.~C., I.~Lobel, and R.~Paes~Leme (2020).
\newblock Feature-based dynamic pricing.
\newblock {\em Management Science\/}~{\em 66\/}(11), 4921--4943.

\bibitem[\protect\citeauthoryear{Cohen, Miao, and Wang}{Cohen
  et~al.}{2021}]{Cohen2021Fairness}
Cohen, M.~C., S.~Miao, and Y.~Wang (2021).
\newblock Dynamic pricing with fairness constraints.

\bibitem[\protect\citeauthoryear{Demirkaya, Fan, Gao, Lv, Vossler, and
  Wang}{Demirkaya et~al.}{2022}]{demirkaya2022optimal}
Demirkaya, E., Y.~Fan, L.~Gao, J.~Lv, P.~Vossler, and J.~Wang (2022).
\newblock Optimal nonparametric inference with two-scale distributional nearest
  neighbors.
\newblock {\em Journal of the American Statistical Association\/}, 1--11.

\bibitem[\protect\citeauthoryear{Den~Boer}{Den~Boer}{2015}]{den2015dynamic}
Den~Boer, A.~V. (2015).
\newblock Dynamic pricing and learning: historical origins, current research,
  and new directions.
\newblock {\em Surveys in operations research and management science\/}~{\em
  20\/}(1), 1--18.

\bibitem[\protect\citeauthoryear{Fan, Guo, and Yu}{Fan
  et~al.}{2022}]{fan2022policy}
Fan, J., Y.~Guo, and M.~Yu (2022).
\newblock Policy optimization using semiparametric models for dynamic pricing.
\newblock {\em Journal of the American Statistical Association\/}, 1--29.

\bibitem[\protect\citeauthoryear{Feng and Zhu}{Feng and
  Zhu}{2023}]{feng2023principro}
Feng, Q. and R.~Zhu (2023, Jul).
\newblock Principro: Data-driven algorithms for joint pricing and inventory
  control under price protection.
\newblock Available at SSRN:4511384.

\bibitem[\protect\citeauthoryear{Guan and Jiang}{Guan and
  Jiang}{2018}]{guan2018nonparametric}
Guan, M. and H.~Jiang (2018).
\newblock Nonparametric stochastic contextual bandits.
\newblock In {\em Proceedings of the AAAI Conference on Artificial
  Intelligence}, Volume~32.

\bibitem[\protect\citeauthoryear{Hu, Kallus, and Mao}{Hu
  et~al.}{2022}]{hu2022smooth}
Hu, Y., N.~Kallus, and X.~Mao (2022).
\newblock Smooth contextual bandits: Bridging the parametric and
  nondifferentiable regret regimes.
\newblock {\em Operations Research\/}~{\em 70\/}(6), 3261--3281.

\bibitem[\protect\citeauthoryear{Javanmard and Nazerzadeh}{Javanmard and
  Nazerzadeh}{2019}]{javanmard2019dynamic}
Javanmard, A. and H.~Nazerzadeh (2019).
\newblock Dynamic pricing in high-dimensions.
\newblock {\em The Journal of Machine Learning Research\/}~{\em 20\/}(1),
  315--363.

\bibitem[\protect\citeauthoryear{Jiang}{Jiang}{2019}]{jiang2019non}
Jiang, H. (2019).
\newblock Non-asymptotic uniform rates of consistency for k-nn regression.
\newblock In {\em Proceedings of the AAAI Conference on Artificial
  Intelligence}, Volume~33, pp.\  3999--4006.

\bibitem[\protect\citeauthoryear{Keskin and Zeevi}{Keskin and
  Zeevi}{2017}]{keskin2017chasing}
Keskin, N.~B. and A.~Zeevi (2017).
\newblock Chasing demand: Learning and earning in a changing environment.
\newblock {\em Mathematics of Operations Research\/}~{\em 42\/}(2), 277--307.

\bibitem[\protect\citeauthoryear{Kleinberg and Leighton}{Kleinberg and
  Leighton}{2003}]{kleinberg2003value}
Kleinberg, R. and T.~Leighton (2003).
\newblock The value of knowing a demand curve: Bounds on regret for online
  posted-price auctions.
\newblock In {\em 44th Annual IEEE Symposium on Foundations of Computer
  Science, 2003. Proceedings.}, pp.\  594--605. IEEE.

\bibitem[\protect\citeauthoryear{Lattimore and Szepesv{\'a}ri}{Lattimore and
  Szepesv{\'a}ri}{2020}]{lattimore2020bandit}
Lattimore, T. and C.~Szepesv{\'a}ri (2020).
\newblock {\em Bandit algorithms}.
\newblock Cambridge University Press.

\bibitem[\protect\citeauthoryear{Lee}{Lee}{2019}]{lee2019u}
Lee, A.~J. (2019).
\newblock {\em U-statistics: Theory and Practice}.
\newblock Routledge.

\bibitem[\protect\citeauthoryear{Li and Zheng}{Li and
  Zheng}{2023}]{li2023dynamic}
Li, X. and Z.~Zheng (2023).
\newblock Dynamic pricing with external information and inventory constraint.
\newblock {\em Management Science\/}~{\em 0\/}(0).

\bibitem[\protect\citeauthoryear{Luo, Sun, and Liu}{Luo
  et~al.}{2024}]{luo2024distribution}
Luo, Y., W.~W. Sun, and Y.~Liu (2024).
\newblock Distribution-free contextual dynamic pricing.
\newblock {\em Mathematics of Operations Research\/}~{\em 49\/}(1), 599--618.

\bibitem[\protect\citeauthoryear{Maestre, Duque, Rubio, and Arévalo}{Maestre
  et~al.}{2018}]{Maestre2018Reinforcement}
Maestre, R., J.~Duque, A.~Rubio, and J.~Arévalo (2018).
\newblock Reinforcement learning for fair dynamic pricing.
\newblock In {\em Proceedings of SAI Intelligent Systems Conference}, pp.\
  120--135. Springer.

\bibitem[\protect\citeauthoryear{Mao, Leme, and Schneider}{Mao
  et~al.}{2018}]{NEURIPS2018}
Mao, J., R.~Leme, and J.~Schneider (2018).
\newblock Contextual pricing for lipschitz buyers.
\newblock In {\em Advances in Neural Information Processing Systems},
  Volume~31.

\bibitem[\protect\citeauthoryear{Nambiar, Simchi-Levi, and Wang}{Nambiar
  et~al.}{2019}]{nambiar2019dynamic}
Nambiar, M., D.~Simchi-Levi, and H.~Wang (2019).
\newblock Dynamic learning and pricing with model misspecification.
\newblock {\em Management Science\/}~{\em 65\/}(11), 4980--5000.

\bibitem[\protect\citeauthoryear{Perchet and Rigollet}{Perchet and
  Rigollet}{2013}]{perchet2013multi}
Perchet, V. and P.~Rigollet (2013).
\newblock The multi-armed bandit problem with covariates.
\newblock {\em The Annals of Statistics\/}~{\em 41\/}(2), 693--721.

\bibitem[\protect\citeauthoryear{Phillips, Şimşek, and van Ryzin}{Phillips
  et~al.}{2015}]{Phillips2015Effectiveness}
Phillips, R., A.~S. Şimşek, and G.~van Ryzin (2015).
\newblock The effectiveness of field price discretion: Empirical evidence from
  auto lending.
\newblock {\em Management Science\/}~{\em 61\/}(8), 1741--1759.

\bibitem[\protect\citeauthoryear{Qiang and Bayati}{Qiang and
  Bayati}{2016}]{qiang2016dynamic}
Qiang, S. and M.~Bayati (2016).
\newblock Dynamic pricing with demand covariates.
\newblock {\em arXiv preprint arXiv:1604.07463\/}.

\bibitem[\protect\citeauthoryear{Shah, Johari, and Blanchet}{Shah
  et~al.}{2019}]{NEURIPS2019}
Shah, V., R.~Johari, and J.~Blanchet (2019).
\newblock Semi-parametric dynamic contextual pricing.
\newblock In {\em Advances in Neural Information Processing Systems},
  Volume~32.

\bibitem[\protect\citeauthoryear{Slivkins}{Slivkins}{2011}]{slivkins2011contextual}
Slivkins, A. (2011).
\newblock Contextual bandits with similarity information.
\newblock In {\em Proceedings of the 24th annual Conference On Learning
  Theory}, pp.\  679--702. JMLR Workshop and Conference Proceedings.

\bibitem[\protect\citeauthoryear{Slivkins et~al.}{Slivkins
  et~al.}{2019}]{slivkins2019introduction}
Slivkins, A. et~al. (2019).
\newblock Introduction to multi-armed bandits.
\newblock {\em Foundations and Trends{\textregistered} in Machine
  Learning\/}~{\em 12\/}(1-2), 1--286.

\bibitem[\protect\citeauthoryear{Steele}{Steele}{2009}]{steele2009exact}
Steele, B.~M. (2009).
\newblock Exact bootstrap k-nearest neighbor learners.
\newblock {\em Machine Learning\/}~{\em 74}, 235--255.

\bibitem[\protect\citeauthoryear{Stone}{Stone}{1982}]{Stone1982}
Stone, C.~J. (1982).
\newblock Optimal global rates of convergence for nonparametric regression.
\newblock {\em Ann. Statist.\/}~{\em 10\/}(4), 1040--1053.

\bibitem[\protect\citeauthoryear{van~den Boer and Keskin}{van~den Boer and
  Keskin}{2022}]{vandenBoer2022dynamic}
van~den Boer, A.~V. and N.~B. Keskin (2022).
\newblock Dynamic pricing with demand learning and reference effects.
\newblock {\em Management Science\/}~{\em 68\/}(10), 7112--7130.

\bibitem[\protect\citeauthoryear{Wang, Wang, Sun, and Cheng}{Wang
  et~al.}{2023}]{Wang2023OnlineRegularization}
Wang, C.~H., Z.~Wang, W.~W. Sun, and G.~Cheng (2023).
\newblock Online regularization toward always-valid high-dimensional dynamic
  pricing.
\newblock {\em Journal of the American Statistical Association\/}, 1--13.

\bibitem[\protect\citeauthoryear{Wang, Chen, Chang, and Ge}{Wang
  et~al.}{2021}]{Wang2021Uncertainty}
Wang, Y., X.~Chen, X.~Chang, and D.~Ge (2021).
\newblock Uncertainty quantification for demand prediction in contextual
  dynamic pricing.
\newblock {\em Production and Operations Management\/}~{\em 30\/}(6),
  1703--1717.

\bibitem[\protect\citeauthoryear{Xu and Wang}{Xu and
  Wang}{2021}]{xu2021logarithmic}
Xu, J. and Y.-X. Wang (2021).
\newblock Logarithmic regret in feature-based dynamic pricing.
\newblock {\em Advances in Neural Information Processing Systems\/}~{\em 34},
  13898--13910.

\end{thebibliography}

%
%
\clearpage
\setstretch{1.68}
\setcounter{page}{1}
\begin{appendices}
	\begin{center}
		{\Large Supplementary Material of \\ ``\TTL''}
		
		\vspace{2em}
		{
			Elynn Chen$^\flat$ \hspace{2ex}
			Xi Chen$^\natural$ \hspace{2ex}
			Lan Gao$^\sharp$ \hspace{2ex}
			Jiayu Li$^\dag$ \\ \normalsize
			\smallskip
			$^{\flat,\natural,\dag}$Leonard N. Stern School of Business, New York University. \\ \normalsize
			\smallskip
			$^{\sharp}$ Haslam College of Business, University of Tennessee Knoxville.
		}
	\end{center}

	\section{Proof of Regret Bounds (Theorems \ref{thm-1} and \ref{thm-TDNN})}
	\label{append:proof-regret}
	The regret upper bound relies on the statistical convergence rate in estimating the mean utility function $\mathring{\mu}(\bx)$.
	To present our proof in a more general framework, Lemmas \ref{lemma-F}, \ref{lemma-F(1)}, \ref{lemma-sup-h} and most proofs of Theorem \ref{le-DNN-TDNN} are build upon the following general assumption on the accuracy for an estimator $\hat{\mu}_k (\bx)$ of the $\mathring{\mu}(\bx)$. 
	Specifically, when dealing with the DNN and TDNN estimators outlined in Section \ref{sec:mean-utility}, we can deduce that the convergence rate $r_{n_{k, exp}} \leq C  n_{k, \exp}^{- \frac{2} {d + 4}} \sqrt{ d \log n_{k, \exp}} $ by applying Theorem \ref{le-DNN-TDNN} with $\delta = 0.5 n_{k, exp}^{-1}$ and $s = O(n_{k, \exp}^{\frac{d}{d +4}})$.
	
	\begin{assumption} \label{assump:f_hat}
		Suppose an estimator $\hat{\mu}_k(\bx)$ is applied using i.i.d. observations $\{(\bx_i, By_i)\}_{i\in\calT_{k, exp}}$, where $\calT_{k, exp}$ is the exploration phase in the $k$-th episode. 
		Let $n_{k, exp} = \abs{\calT_{k, exp}}$ be the number of observations in $\calT_{k, exp}$.
		Assume  there exists a vanishing  sequence $r_{n_{k, exp}}$ such that
		\begin{equation}
			\PP\paren{ \sup_{\bx  \in \mathcal{X} } | \hat{\mu}_k (\bx) - \mathring{\mu}(\bx) |  \geq r_{n_{k, exp}} }  \leq n_{k, exp}^{-1}.
		\end{equation}
	\end{assumption}
	
	\subsection{Proof of Theorem \ref{thm-1}}
	The proof consists of two major steps. The first step is to establish the approximation accuracy of the kernel estimator $\hat{F}_k (z; \hat{\mu}_k)$ for estimating the noise distribution $F$, provided with the accuracy of the estimator $\hat{\mu}_k$ for the non-parametric regression function $\mathring{\mu}$. 
	Therefore, we can obtain the approximation accuracy of the data-driven optimal price $p_t$ towards the optimal price $p_t^*$. The second step is built on the accuracy derived for $p_t$ in the first step and establishes the overall regret bound. 
	For simplicity, we denote $\hat{F}_k (z; \hat{\mu}_k) := \hat{F}_k (z; \hat{\mu}_k, \{\bx_t, y_t, p_t\}_{t \in \cT_{k, exp}})$ by ignoring its dependence on the data points $\{\bx_t, y_t, p_t\}_{t \in \cT_{k, exp}}$ collected in the exploration phase.
	We define 
	\begin{equation} \label{def_N_k}
		\calN_k = \{\mu:\|\mu-\mathring{\mu}\|_{\infty} := \sup_{\bx \in \mathcal{X}} |\mu(\bx) - \mathring{\mu}(\bx)| \le r_{n_{k, exp}}\},
	\end{equation}
	and choose the bandwidth $b_k = n_{k, exp}^{\frac{1}{2m + 1}}$ in the kernel estimators defined in \eqref{eqn:aux-def-a-xi}.
	
	\noindent
	{\it Step 1. Establish the following three lemmas on the approximation accuracy of kernel estimator $\hat{F}_k$ and its first-order derivative. In addition, we will also show the approximation accuracy of the estimator $\hat{h}_k(\cdot)$ (defined in \eqref{eqn:func-h-hat}) which plays a crucial role in determining the data-driven optimal price $p_t$.}  
	Their proofs are presented in Sections \ref{sec:lemma9-pf}, \ref{sec:lemma10-pf} and \ref{sec:lemma11-pf}, respectively.
	
	\begin{lemma}[Generalization of Theorem \ref{lemma-F-DNN}] \label{lemma-F}
		Under Assumptions \ref{assump:F}--\ref{assump:regu-phi} and \ref{assump:f_hat}, there exist constants $C_1> 0$ and $C_2 > 0$ such that for each $k$ satisfying $n_{k, exp}^{\frac{2m}{2m + 1}} \geq \max(C_1 d ( \frac{\log n_{k, exp}}{2m + 1} + 1 ), 3 \log n_{k, exp}) $ and each $\delta \in \big(4 \exp \{- C_2 n_{k, exp}^{\frac{m}{2m + 1}} / \log n_{k, exp} \}, \, \frac{1} {2} \big]$, it holds with probability at least $ 1 - 2\delta$ that
		\begin{equation} \label{eqn:kernel-accuracy}
			\sup_{z \in \calS_{\eps}, {\mu}\in \mathcal{N}_{k}} \abs{ \hat{F}_k (z; {\mu}) - F(z) } 
			\leq  B_F n_{k, exp}^{- \frac{m}{2m + 1}} \sqrt{\log n_{k, exp} } ( \sqrt d + \sqrt{\log  \delta^{-1}} ) + B_F r_{n_{k, exp}},
		\end{equation}
		where $B_F > 0$ is a constant. 
	\end{lemma}

	\begin{lemma}[First derivative of noise c.d.f.] \label{lemma-F(1)}
		Under the same condition as Lemma \ref{lemma-F}, it holds with probability at least $ 1 - 4\delta$ that 
		\begin{equation} \label{eqn:kernel-accuracy-deriva}
			\sup_{z \in \calS_{\eps}, {\mu}\in \mathcal{N}_{k}} \abs{ \hat{F}_k^{(1)} (z; {\mu}) - F^{(1)}(z) } \leq  B_f n_{k, exp}^{- \frac{m-1}{2m + 1}} \sqrt{\log n_{k, exp} } ( \sqrt d + \sqrt{\log \delta^{-1}} ) + B_f r_{n_{k, exp}},
		\end{equation}
		where $B_f > 0$ is a constant. 
	\end{lemma}
	
	\begin{lemma} \label{lemma-sup-h}
		Assume the same condition as Lemma \ref{lemma-F}, and $ n_{k, exp}^{\frac{m  -1} {2m + 1}} \gtrsim  \sqrt{\log n_{k, exp} } (\sqrt d + \sqrt {\log \delta^{-1}}) $ and $r_{n_{k, exp}} \lesssim 1$, it holds with probability at least $ 1 - 6\delta$ that
		\begin{equation}
			\underset{z \in[B_{\eps} - B, B_{\eps}]}{\sup}\abs{\hat{h}_k(z)-h(z)} 
			\leq B_h n_{k, exp}^{- \frac{m-1}{2m + 1}} \sqrt{\log n_{k, exp} } ( \sqrt d + \sqrt{\log \delta^{-1}} ) + B_h r_{n_{k, exp}},
		\end{equation}
		where $B_h > 0$ is a constant. 
	\end{lemma}
	
	\noindent
	{\it Step 2. Establish the regret bound.}
	We will apply Theorem \ref{le-DNN-TDNN} and Lemmas \ref{lemma-F}, \ref{lemma-F(1)}, and \ref{lemma-sup-h}.
	
	We are now ready to bound the regret of our policy. 
	For a policy $\pi$ that assigns a price $p_t(\pi)$ at $t$, its instant regret at time $t$ is defined as
	\begin{equation}
		\reg_t(\pi) \defeq p_t^* \bbone(v_t \geq p_t^*) - p_t(\pi) \bbone(v_t \geq p_t(\pi)),
	\end{equation}
	where $p_t^*$ is the oracle optimal price given in \eqref{eqn:func-h}.
	For the episode-based algorithm, we further define the accumulated regret in each episode $\mathcal{T}_k$:
	\begin{equation}
		\reg_{(k)}(\pi) = \sum_{t\in\calT_{k}}\reg_t(\pi).
	\end{equation}
	Since the length of the episodes grows exponentially, the number of episodes by period $T$ is logarithmic in $T$. 
	Specifically, $T$ belongs to episode $K = \floors{\log_2 (T/n_0)} + 1$.
	Without loss of generality, we deal with the setting $T = 2^{K-1}$ for some integer number $K$. 
	Hence, regret over the time horizon of $T$ can be expressed as
	\begin{equation}
		\Reg(T;\pi) = \EE\brackets{\sum_{k = 1}^K \reg_{(k)}(\pi)}.
	\end{equation}
	We bound the total regret over each episode $k$ by considering two separate cases where $k$ is relatively small or large.
	
	\noindent
	\underline{\textsc{Case I.}}
	When $k \leq 1 + (\log \sqrt T - \log n_0 ) / \log 2 $, we have  $n_{k} = 2^{k-1} n_0 \leq \sqrt T $. Those episodes are not long enough to accurately estimate $\mathring{\mu}$ and $F$. 
	We use the fact that $\EE\brackets{\reg_t(\pi)} \le p_t^* \le B$ to construct a naive upper bound. 
	That is, 
	\begin{equation}  \label{eq_small_k_bound}
		\begin{aligned}
			\sum_{k \leq 1 + (\log \sqrt T - \log n_0 ) / \log 2  } \reg_{(k)}(\pi) & =  \sum_{k \leq 1 + (\log \sqrt T - \log n_0 ) / \log 2  }   B  2^{k-1} n_0\\
			& \le   B 2^{1 + (\log \sqrt T - \log n_0 ) / \log 2  } n_0  \\
			& \leq 2 n_0 B e^{\log \sqrt T - \log n_0} = 2 B \sqrt T.
		\end{aligned}
	\end{equation}
	
	\noindent
	\underline{\textsc{Case II.}}
	When  $k > 1 + (\log \sqrt T - \log n_0 ) / \log 2 $, we consider 
	\begin{equation} \label{eqn:E[reg_t]}
		\EE\brackets{\reg_t(\pi)} = \EE\brackets{\EE\brackets{\reg_t(\pi)\cond \calH_{t-1}^+}}
	\end{equation}
	where $\calH_{t-1}^+\defeq \sigma\braces{\bx_1,\cdots,\bx_{t-1}, \bx_{t}; \varepsilon_1, \cdots, \varepsilon_{t-1}}$ is the filtration generated by $\braces{\bx_1,\cdots,\bx_{t-1}; \varepsilon_1, \cdots, \varepsilon_{t-1}}$ and the current feature $\bx_t$.
	We write
	\begin{equation} \label{eqn:E[reg_t|H_t-1]}
		\begin{aligned}
			\EE\brackets{\reg_t(\pi) \cond \calH_{t-1}^+}
			& = \EE\brackets{p_t^* \bbone(v_t\geq p_t^*) \cond \calH_{t-1}^+} 
			- \EE\brackets{p_t \bbone(v_t\geq p_t) \cond \calH_{t-1}^+}  \\
			& = \EE\brackets{p_t^* \bbone(\varepsilon_t + \mathring{\mu}(\bx_t) \geq p_t^*) \cond \calH_{t-1}^+} \\
			&- \EE\brackets{p_t \bbone(\varepsilon_t + \mathring{\mu}(\bx_t) \geq p_t) \cond \calH_{t-1}^+} \\
			& = p_t^* \paren{1 - F(p_t^*-\mathring{\mu}(\bx_t))} 
			- p_t \paren{1 - F(p_t-\mathring{\mu}(\bx_t))} \\
			& \defeq \rev_t(p_t^*) - \rev_t(p_t), 
		\end{aligned}
	\end{equation}
	where $\rev_t(p)\defeq p \paren{1 - F(p-\mathring{\mu}(\bx_t))}$ is defined as the expected revenue under price $p$. 
	Note that $p_t^*\in\argmax_p\;\rev_t(p)$ and thus $\rev_t^{(1)}(p_t^*) = 0$. It follows by Taylor expansion that
	\begin{equation} \label{eqn:rev_t_decomp}
		\rev_t(p_t) = \rev_t(p_t^*) + \frac{1}{2}\rev_t^{(2)}(p)(p_t-p_t^*)^2, 
	\end{equation}
	for some $p$ between $p_t$ and $p_t^*$. 
	Let $C = \max\{\max_z F^{(1)}(z),\;\max_u F^{(2)}(z)\}$,
	we claim that, for any $p$ between $p_t$ and $p_t^*$, 
	\begin{equation} \label{eqn:rev_t^{(2)}}
		\abs{\rev_t^{(2)}(p)} \le (2+B)C, 
	\end{equation}
	which is derived by the fact that 
	\begin{equation*}
		\abs{\rev_t^{(2)}(p)} = \abs{-2F^{(1)}(p-\mathring{\mu}(\bx_t))-pF^{(2)}(p-\mathring{\mu}(\bx_t))},
	\end{equation*}
	and $p_t, p_t^* \in [0,B]$.
	Combining \eqref{eqn:E[reg_t|H_t-1]}, \eqref{eqn:rev_t_decomp} and \eqref{eqn:rev_t^{(2)}}, we obtain
	\begin{align*}
		\EE\brackets{\reg_t(\pi) \cond \calH_{t-1}^+}
		& \le (2+B)C/2\cdot \EE\brackets{(p_t(\pi)-p_t^*)^2}
	\end{align*}
	Let $\hat{p}_t := p_t(\hat{\pi})$ be the price given by Algorithm \ref{algo:non-par-pricing}. Then it holds that
	\begin{align*}
		\EE\brackets{\reg_t(\hat\pi) \cond \calH_{t-1}^+}
		& \le (2+B)C/2\cdot \EE\brackets{(\hat{p}_t-p_t^*)^2 }. 
	\end{align*}
	We can obtain by \eqref{eqn:E[reg_t]} as well as the definitions of $p_t^*$ and $\hat{p_t} $ in \eqref{eqn:opt-p} and \eqref{eqn:opt-price-est} that
	\begin{align}  
		\EE\brackets{\reg_t(\hat{\pi})} & \leq 
		{(2+B)C}/{2} \cdot \EE\brackets{\paren{ \hat{h}(\hat{\mu}_k(\bx_t)) - h(\mathring{\mu}(\bx_t)) }^2} \nonumber \\
		& \le (2 + B) C \Big\{ \EE\brackets{\paren{\hat{h}(\hat{\mu}_k(\bx_t))-h(\hat{\mu}_k(\bx_t))}^2} \\
		&+ \EE\brackets{\paren{h(\hat{\mu}_k(\bx_t))-h(\mathring{\mu}(\bx_t))}^2} \Big\}. \label{eqn:E[reg_t]-decomp}
	\end{align}
	We first analyze the first term on the right hand side of \eqref{eqn:E[reg_t]-decomp}. It can be seen from Lemma \ref{thm:g*(x)-bound} that for $\hat{\mu}_k \in \mathcal{N}_k$, 
	\begin{align*}
		\paren{\hat{h}(\hat{\mu}_k(\bx_t))-h(\hat{\mu}_k(\bx_t))}^2
		& \le \paren{ \underset{z\in[B_{\varepsilon}, B-B_{\varepsilon}]}{\sup}\paren{\hat{h}_k(z)-h(z)}}^2. 
	\end{align*}
	Furthermore, applying Lemma \ref{lemma-sup-h} (let $\delta = n_{k, exp}^{-1}$), we have 
	\begin{equation}
		\begin{aligned}
			& \EE\brackets{\paren{\hat{h}(\hat{\mu}_k(\bx_t))-h(\hat{\mu}_k(\bx_t))}^2 }  \\
			& \le  \paren{ B_h n_{k, exp}^{- \frac{m-1}{2m + 1}} \sqrt{\log n_{k, exp} } ( \sqrt d + \sqrt{\log n_{k, exp}} ) + B_h r_{n_{k, exp}} }^2 + B^2  n_{k, exp}^{-1}.
		\end{aligned}
	\end{equation}
	Next, for the second term in \eqref{eqn:E[reg_t]-decomp}, 
	in view of the fact that $\frac{d} {d z}\phi^{-1} (z) = \frac{1} {\phi^{(1)} (\phi^{-1}(z)) }$, we can obtain by Assumption \ref{assump:regu-phi} that
	\begin{equation}
		\begin{aligned}
			2 \EE\brackets{\big(h(\hat{\mu}_k(\bx_t))-h(\mathring{\mu}(\bx_t)) \big) ^2  }
			& = 2 \EE\brackets{ \big( {\phi}^{-1} (\hat{\mu}_k(\bx_t) ) -  {\phi}^{-1} (\mathring{\mu}(\bx_t)) \big)^2 } \\
			& \le 2 L_{\phi}^{-2} \EE \big[ \paren{\hat{\mu}_k(\bx_t) - \mathring{\mu}(\bx_t)}^2 \big].
		\end{aligned}
	\end{equation}
	In addition, under Assumption \ref{assump:f_hat}, it holds that
	\begin{equation} \label{eqn-second-part}
		\begin{aligned}
			\EE \big[ \paren{\hat{\mu}_k(\bx_t) - \mathring{\mu}(\bx_t)}^2 \big] & \leq \EE \big[ \paren{\hat{\mu}_k(\bx_t) - \mathring{\mu}(\bx_t)}^2 \mathbbm{1} (  \sup_{\bx  \in \mathcal{X} } \| \hat{\mu}_k (\bx) - \mathring{\mu}(\bx) \|   < r_{n_{k, exp}} )\big] \\
			& \quad + B^2  \EE \big[ \mathbbm{1} (  \sup_{\bx  \in \mathcal{X} } \| \hat{\mu}_k (\bx) - \mathring{\mu}(\bx) \|   \geq  r_{n_{k, exp}} )\big]  \\
			& \le  r_{ n_{k, exp}}^2 +  B^2 n_{k, exp}^{-1}   .
		\end{aligned}
	\end{equation} 
	Combining the inequalities \eqref{eqn:E[reg_t]-decomp} -- \eqref{eqn-second-part} yields the following upper bound for the expected regret at any time $t$ during the commitment phase $\mathcal{T}_{k, com}$ in episode $k$:
	\begin{align*}
		\EE\brackets{\reg_t(\pi)} 
		& \le (2 + B) C \big( B_h n_{k, exp}^{- \frac{m-1}{2m + 1}} \sqrt{\log n_{k, exp} } ( \sqrt d + \sqrt{\log n_{k, exp}} ) \\
		&+ (B_h +1) r_{n_{k, exp}} \big)^2 +  2C( 2 + B) B^2 n_{k, exp}^{-1}.
	\end{align*} 
	Consequently, for $k > 1 + (\log \sqrt T - \log n_0)/\log 2$, the total regret during the $k$-th episode is 
	\begin{equation} \label{eq_k_large_bound}
		\begin{aligned}
			\Reg_{(k)}(\pi)
			& = \sum_{t\in \calT_{k,exp}}\EE\brackets{\reg_t(\pi)} + \sum_{t\in \calT_{k,com}}\EE\brackets{\reg_t(\pi)}  \\
			& \le B n_{k, exp} + n_{k,com} \EE\brackets{\reg_t(\pi)} \\
			&   \le B n_{k, exp} + n_{k, com}  (2 + B) C \Big(B_h n_{k, exp}^{- \frac{m-1}{2m + 1}} \sqrt{\log n_{k, exp} } ( \sqrt d + \sqrt{\log n_{k, exp}} ) \\
			&+ (B_h +1 ) r_{n_{k, exp}}\Big)^2 + 2 C (2 + B) B^2 n_{k, com} n_{k, exp}^ {-1  }.
		\end{aligned}
	\end{equation}
	
	Now we need to consider specific estimators with specific convergence rate $r_{n_{k, exp}}$ in order to further choose an optimized exploration phase size and commitment phase size to minimize the regret $\Reg_{(k)}(\pi)$. We consider DNN estimator in the sequel. 
	
	When DNN estimator is applied to estimate the mean utility function, we choose the involved subsampling scale $s = n_{k, exp}^{\frac{d}{d + 4}}$, then it follows from Theorem \ref{le-DNN-TDNN} that Assumption \ref{assump:f_hat} is satisfied for the DNN estimator with convergence rate $r_{n_{k, exp} } = C n_{k, exp}^{\frac{-2}{d+4}} \sqrt{d\log n_{k, exp}}$. 
	If $d \geq \frac{6}{m - 1}  $, $ n_{k, exp}^{\frac{-2}{d+4}} \geq n_{k, exp}^{ - \frac{m - 1}{2m+1}} $, otherwise  if $1 \leq d < \frac{6}{m-1}$, $n_{k, exp}^{\frac{-2}{d+4}} < n_{k, exp}^{- \frac{m-1}{2m + 1}} $. When $1 \leq d < \frac{6}{m-1}$, by setting $n_{k, exp} = n_k^{\frac{2m +1}{4m -1 }} $ and $n_k = 2^{k-1} n_0$, we have $n_{k, exp} \geq T^{\frac{m} {4m - 1}} $ for $k > 1 + (\log \sqrt T - \log n_0)/\log 2$, and hence the condition $ n_{k, exp}^{\frac{2m}{2m + 1}} \geq \max(C_1 d ( \frac{\log n_{k, exp}}{2m + 1} + 1 ), 3 \log n_{k, exp})  $ is satisfied for Lemmas \ref{lemma-F} -- \ref{lemma-sup-h}. 
	Noting that $n_{k, exp}^{-1/2} = o( r_{n_{k, exp}})$,  it follows from \eqref{eq_k_large_bound} that for $k > 1 + (\log \sqrt T - \log n_0)/\log 2$, 
	\begin{equation}  \label{eq_bound_k_large_new}
		\begin{aligned}
			\Reg_{(k)}(\pi)
			& \le B n_{k, exp} +   C n_{k, com}  \paren{ n_{k, exp}^{- \frac{m-1}{2m + 1}} \sqrt{\log n_{k, exp} } ( \sqrt d + \sqrt{\log n_{k, exp}} ) }^2 \\
			& \leq B n_k^{\frac{2m +1}{4m -1 }} +  C   n_{k}^{\frac{2m +1}{4m -1 }} [(\log n_k)^2  + d \log n_k] \\
			& \leq C n_{k}^{\frac{2m +1}{4m -1 }} [(\log n_k)^2  + d \log n_k],
		\end{aligned}
	\end{equation}
	where $C$ is a constant which can take different values from line to line. In view of \eqref{eq_small_k_bound} and \eqref{eq_bound_k_large_new}, we obtain that the total regret  over time horizon $T$ is bounded by
	\begin{equation} \label{small_d_bound1}
		\begin{split}
			\Reg(\pi, T) & = \sum_{k=1}^{\log_2 (T/n_0) + 1 } \Reg_k(\pi) \\
			& \le 2 B \sqrt T  + \sum_{k = 1 + (\log \sqrt T - \log n_0)/\log 2}^{\log_2 (T/n_0) + 1} C n_{k}^{\frac{2m +1}{4m -1 }} [(\log n_k)^2  + d \log n_k] \\
			& \leq  2 B \sqrt T  + C n_0^{\frac{2m +1}{4m -1 }} [(\log T)^2  + d \log T] \sum_{k = 1}^{\log_2 (T/n_0) + 1} 2^{(k-1) \frac{2m +1}{4m -1 }} \\
			& \leq  2 B \sqrt T  +  C  n_0^{\frac{2m +1}{4m -1 }} [(\log T)^2  + d \log T] (2T/n_0)^{\frac{2m+1}{4m-1}} \\
			& \leq  C T^{\frac{2m+1}{4m-1}}   [(\log T)^2  + d \log T] ,
		\end{split}
	\end{equation}
	where we have used the definition that $n_k= 2^{k-1} n_0$ and the fact that $\sqrt T = o(T^{\frac{2m + 1}{4 m - 1}})$. 
	
	When $d \geq \frac{6}{m-1}$, by letting $n_{k, exp} = n_{k}^{\frac{d + 4}{d + 8}}$, we have for $k > 1 + (\log \sqrt T -
	log n_0)/\log_2$,
	\begin{align*}
		\Reg_{(k)}(\pi)
		& \le B n_{k, exp} +     C n_{k, com}  \paren{   n_{k, exp}^{\frac{-2}{d+4}} \sqrt{\log n_{k, exp} } ( \sqrt d + \sqrt{\log n_{k, exp}} ) }^2 \\
		& \le C n_{k}^{\frac{d + 4}{d + 8}} ( d\log n_{k} + (\log n_{k})^2 )
	\end{align*}
	Therefore, combining \eqref{eq_small_k_bound} and \eqref{eq_bound_k_large_new}, the total regret over time horizon $T$ is bounded by
	\begin{equation} \label{large_d_bound2}
		\begin{split}
			\Reg(\pi, T) & = \sum_{k=1}^{\log_2 (T/n_0) + 1} \Reg_k(\pi) \\
			& \le 2 B \sqrt T  + \sum_{k=1}^{\log_2 (T/n_0) + 1} C n_{k}^{\frac{d + 4}{d + 8}} [(\log n_k)^2  + d \log n_k] \\
			& \leq  2 B \sqrt T  +  \sum_{k=1}^{\log_2 (T/n_0) + 1} C (2^{k-1} n_{0})^{\frac{d + 4}{d + 8}} [(\log T)^2  + d \log T]  \\
			& \leq 2 B \sqrt T  +   C  n_{0}^{\frac{d + 4}{d + 8}} (2T/n_0)^{\frac{d + 4}{d + 8}} [(\log T)^2  + d \log T]  \\
			& \leq C T^{\frac{d + 4}{d + 8}}   [(\log T)^2  + d \log T].
		\end{split}
	\end{equation}
	Combining the bound in \eqref{small_d_bound1} as $d < \frac{6}{m-1}$ and the bound in \eqref{large_d_bound2} as $d \geq \frac{6}{m-1}$ yields the desired result in \eqref{result-thm-DNN-regret}. The proof of Theorem \ref{thm-1} is completed.

	\subsection{Proof of Theorem \ref{thm-TDNN}} \label{pf-thm-TDNN}
	When TDNN estimator is applied in the mean utility function estimation with $s_1 = n_{k, exp}^{\frac{d}{d + 8} \lor \frac{1}{7}}$, it follows from Theorem \ref{le-DNN-TDNN} that the uniform convergence rate ${n_{k, exp}^{- (\frac{4} {d + 8} \land \frac{3}{7}) } \sqrt {d \log n_{k, exp}}}$ holds in the Assumption \ref{assump:f_hat}. 
	If $m  \geq 1 +  \min(\frac{12} {d}, 9)   $, it holds that $n_{k, exp}^{- \frac{4}{d+8} \land \frac{3}{7}} \geq n_{k, exp}^{- \frac{m-1}{2m + 1}} $, otherwise  if $ m  < 1 +  \min(\frac{12} {d}, 9) $, we have $n_{k, exp}^{- \frac{4}{d+8} \land \frac{3}{7}} < n_{k, exp}^{- \frac{m-1}{2m + 1}} $. When $m  < 1 +  \min(\frac{12} {d}, 9) $, by setting $n_{k, exp} = n_k^{\frac{2m +1}{4m -1 }} $, we have for $k > 1 + (\log \sqrt T - \log n_0)/\log 2$,
	\begin{align*} 
		\Reg_{(k)}(\pi)
		& \le B n_{k, exp} +   C n_{k}  \paren{  n_{k, exp}^{- \frac{m-1}{2m + 1}} \sqrt{\log n_{k, exp} } ( \sqrt d + \sqrt{\log n_{k, exp}} ) }^2 \\
		& \leq B n_k^{\frac{2m +1}{4m -1 }} + C n_{k}^{\frac{2m +1}{4m -1 }} [(\log n_k)^2  + d \log n_k] \\
		& \leq C n_{k}^{\frac{2m +1}{4m -1 }} [(\log n_k)^2  + d \log n_k].
	\end{align*}
	In view of \eqref{eq_small_k_bound} and \eqref{eq_bound_k_large_new}, we obtain that the total regret  over time horizon $T$ is bounded by 
	\begin{equation} \label{small_m_bound1}
		\begin{split}
			\Reg(\pi, T) & = \sum_{k=1}^{\log_2 (T/n_0) + 1} \Reg_k(\pi) \\
			& \le 2 B \sqrt T  + \sum_{k = 1 + (\log \sqrt T - \log n_0 ) / \log 2 }^{\log_2 (T/n_0) + 1}  C n_{k}^{\frac{2m +1}{4m -1 }} [(\log n_k)^2  + d \log n_k] \\
			& \leq  2 B \sqrt T  + C  [(\log T)^2  + d \log T] \sum_{k = 1}^{\log_2 (T/n_0) + 1} (2^{(k-1)} n_0 )^ {\frac{2m +1}{4m -1 }} \\
			& \leq  2 B \sqrt T  +  C  [(\log T)^2  + d \log T] T^{\frac{2m+1}{4m-1}} \\
			& \leq {C} T^{\frac{2m+1}{4m-1}}  [(\log T)^2  + d \log T] ,
		\end{split}
	\end{equation}
	where $C $ is a constant which can take different values from line to line. 
	
	When $m  \geq 1 +  \min(\frac{12} {d}, 9) $, by letting $n_{k, exp} = n_{k}^{\frac{d + 8}{d + 16} \lor \frac{7}{13}}$, we have for $k > 1 + (\log \sqrt T - \log n_0)/\log 2$,
	\begin{align*}
		\Reg_{(k)}(\pi)
		& \le B n_{k, exp} +   C n_{k}  \paren{   n_{k, exp}^{- \frac{4}{d+8} \land \frac{3}{7}} \sqrt{\log n_{k, exp} } ( \sqrt d + \sqrt{\log n_{k, exp}} ) }^2 \\
		& \le C n_{k}^{\frac{d + 8}{d + 16} \lor \frac{7}{13}} ( d\log n_{k} + (\log n_{k})^2 )
	\end{align*}
	Therefore, in view of \eqref{eq_small_k_bound} and \eqref{eq_bound_k_large_new},  the total regret  over time horizon $T$ is bounded by 
	\begin{equation} \label{large_m_bound2}
		\begin{split}
			\Reg(\pi, T) & = \sum_{k=1}^{\log_2 (T/n_0) + 1} \Reg_k(\pi) \\
			& \le 2 B \sqrt T  + C \sum_{k = 1 + (\log \sqrt T 
				- \log n_0)/\log 2}^{\log_2 (T/n_0) + 1}    n_{k}^{\frac{d + 8}{d + 16} \lor \frac{7}{13}} [(\log n_k)^2  + d \log n_k] \\
			& \leq  2 B \sqrt T  + C   [(\log T)^2  + d \log T] \sum_{k = 1}^{\log_2 (T/n_0) + 1} (2^{(k-1)} n_0)^{\frac{d + 8}{d + 16} \lor \frac{7}{13}} \\
			& \leq  2 B \sqrt T  +  C T^{\frac{d + 8}{d + 16} \lor \frac{7}{13}}   [(\log T)^2  + d \log T] \\
			& \leq C   T^{\frac{d + 8}{d + 16} \lor \frac{7}{13}}   [(\log T)^2  + d \log T] ,
		\end{split}
	\end{equation}
	where $C $ is a constant which can take different values from line to line. Combining the bound in \eqref{small_m_bound1} as $m  <  1 +  \min(\frac{12} {d}, 9) $ and the bound in \eqref{large_m_bound2} as $m  \geq 1 +  \min(\frac{12} {d}, 9) $ yields the desired result in \eqref{result-thm-TDNN-regret}. 
	This completes the proof of Theorem \ref{thm-TDNN}.

	\section{Proof of Lemmas}
	
	\subsection{Proof of Lemma $\ref{lemma-F}$} \label{sec:lemma9-pf}
	Recall that $F(u;\mu) \defeq 1 - \EE [y_t\cond p_t - \mu(\bx_t) = u]$ and 
	$F (u) \defeq F(u;\mathring{\mu}) = 1 - \EE [y_t\cond p_t - \mathring{\mu}(\bx_t) = u ]$. 
	The proof builds on the intuition that the kernel estimator $\hat{F}_k (z; \mu)$ converges to the population conditional expectation $F(z; \mu)$ and further $F(z; \mu)$ approaches $F (z)$ as $\mu$ approaches $\mathring{\mu}$.
	By the triangular inequality, we have
	\begin{equation} \label{le3.4-n1}
		\begin{aligned}
			\sup_{z \in \calS_{\eps}, \mu \in \mathcal{N}_k} \abs{ \hat{F}_k (z; \mu) - F(z) }
			&\le 
			\sup_{z \in \calS_{\eps}, \mu \in \mathcal{N}_k} \abs{ \hat{F}_k (z; \mu) - F(z; \mu) } \\
			&+ 
			\sup_{z \in \calS_{\eps}, \mu \in \mathcal{N}_k} \abs{ F(z; \mu) - F (z) }.
		\end{aligned}
	\end{equation}
	In the sequel, we will bound each of the two terms on the right hand side of \eqref{le3.4-n1}.  
	
	First, we concern ourselves with the accuracy of kernel estimator $\hat{F}_k(u; \mu)$ for a fixed  $\mu \in \mathcal{N}_k $. 
	Applying the same technique of proving Lemma 4.2 in \cite{fan2022policy}, we can obtain that with probability at least $1 - 2\delta$,
	\begin{equation} \label{le3.4-1}
		\sup_{z \in \calS_{\eps},  {\mu} \in \mathcal{N}_k} \abs{ \hat{F}_k (z; {\mu}) - F(z;  {\mu} ) } \leq  C  n_{k, exp}^{-\frac{m}{2m + 1}} \sqrt{\log n_{k, exp}} \big(\sqrt d + \sqrt{\log  \delta^{-1} }\Big). 
	\end{equation}
	Next, regarding the second term on the right hand side of \eqref{le3.4-n1}, recalling the definition of the neighborhood $\mathcal{N}_k$ given in \eqref{def_N_k}, we have for any $\mu  \in \mathcal{N}_k $, 
	\begin{equation}
		\begin{aligned}
			F(z;{\mu}) & = 1-  \PP(\mathring{\mu}(\bx_t) + \eps_t \geq p_t \cond  p_t -  {\mu}({\bx}_t) = z) \\
			& = 1-  \PP( \mathring{\mu} (\bx_t) -  {\mu} (\bx_t) + \eps_t \geq p_t-\mu(\bx_t) \cond p_t - {\mu}({\bx}_t) = z) \\
			& \geq 1 - \PP(r_{n_{k, exp}} + \eps_t \geq p_t-{\mu}({\bx}_t) \cond p_t - {\mu}({\bx}_t) = z) \\
			& = F (z - r_{n_{k, exp}}), 
		\end{aligned}    
	\end{equation}
	and similarly, 
	\begin{equation}
		\begin{aligned}
			F(z;\mu) & = 1-\PP(\mathring{\mu}(\bx_t) + \eps_t \geq p_t \cond  p_t - {\mu}({\bx}_t)=z) \\
			& = 1-\PP(\mathring{\mu}(\bx_t)-\mu(\bx_t) + \eps_t \geq p_t-\mu(\bx_t) \cond p_t-\mu({\bx}_t)=z) \\
			& \leq 1-\PP(-r_{n_{k, exp}} + \eps_t\geq p_t-\mu({\bx}_t) \cond p_t-\mu({\bx}_t)=z) \\
			& = F (z + r_{n_{k, exp}}).
		\end{aligned}    
	\end{equation}
	Therefore, by Assumption \ref{assump:F} that $F(u)$ is $\ell_F$-Lipschitz, we have 
	\begin{equation} \label{le3.4-2}
		\sup_{z \in \mathcal{S}_{\varepsilon},  {\mu} \in \mathcal{N}_k} \abs{ F(u;  {\mu}) - F (u) } \leq \ell_{F} r_{n_{k, exp}}. 
	\end{equation}
	A combination of \eqref{le3.4-n1}, \eqref{le3.4-1} and \eqref{le3.4-2} yields  Lemma \ref{lemma-F}. This concludes the proof of Lemma \ref{lemma-F}.

	\subsection{Proof of Lemma $\ref{lemma-F(1)}$} \label{sec:lemma10-pf}
	By the triangular inequality, we have
	\begin{equation} \label{eqn:lemma-F(1)-1}
		\begin{aligned}
			\sup_{z \in \calS_{\eps}, \mu \in \mathcal{N}_k} \abs{ \hat{F}^{(1)}_k (z; \mu) - F^{(1)} (z) }
			&\le 
			\sup_{z \in \calS_{\eps}, \mu \in \mathcal{N}_k} \abs{ \hat{F}_k^{(1)}(z; \mu) - F^{(1)}(z; \mu)} \\
			&+ 
			\sup_{z \in \calS_{\eps}, \mu \in \mathcal{N}_k} \abs{ F^{(1)}(z; \mu) - F^{(1)} (z) }.
		\end{aligned}
	\end{equation}
	Recall that $\xi(z ;{\mu}) $ is the density function of $p_t - \mu(\bx_t)$ evaluated at $z$, $F(z; \mu) \defeq 1 - \EE[y_t \cond p_t - \mu(\bx_t)=z]$, and  $a(z; \mu)\defeq \xi (z; \mu) [1 - F(z; \mu)]$. 
	We also define the following quantities:
	\begin{equation*}
		\begin{aligned}
			\hat a_k(z; \mu) &\defeq \abs{n_{k, exp} b_k}^{-1}\sum_{t \in \calT_{k, exp}} K \paren{\frac{p_t - \mu(\bx_t) - z}{b_k}} y_t, \\
			\hat \xi_k(z; \mu) & \defeq \abs{n_{k, exp} b_k}^{-1}\sum_{t \in \calT_{k, exp}} K\paren{\frac{p_t - \mu(\bx_t) - z }{b_k}}. 
		\end{aligned}
	\end{equation*}
	Then it holds that
	\begin{equation*}
		\hat{F}_k(z; \mu) = 1 - \frac{\hat a_k(z; \mu)}{\hat \xi_k(z; \mu)},\quad
		F(z; \mu) = 1 - \frac{a(z; \mu)}{\xi(z; \mu)}, \quad
		F(z) = 1 - \frac{a(z; \mathring{\mu})}{\xi(z; \mathring{\mu})}, 
	\end{equation*}
	\begin{align*}
		\hat{F}^{(1)}_k(z; \mu) & = - \frac{\hat a_k^{(1)}(z;\mu) \hat \xi_k(z; \mu) - \hat a_k(z; \mu) \hat \xi_k^{(1)}(z; \mu)}{\hat \xi_k(z; \mu)^2}, \\
		F^{(1)}(z; \mu) & = \frac{a^{(1)}(z; \mu)\xi(z; \mu) - a(z; \mu)\xi^{(1)}(z; \mu)}{\xi(z; \mu)^2},
	\end{align*}
	and $F^{(1)} (z) = F^{(1)}(z; \mathring{\mu})$, where
	\begin{align*}
		&\hat a_k^{(1)}(z; \mu) \defeq \frac{-1}{n_{k, exp} b_k^2} \sum_{t \in \calT_{k, exp}} K^{(1)}\paren{\frac{p_t - \mu(\bx_t) - z}{b_k}} y_t, \\
		&\hat \xi_k^{(1)}(z; \mu) \defeq \frac{-1}{n_{k, exp} b_k^2}\sum_{t \in \calT_{k, exp}} K^{(1)}\paren{\frac{p_t - \mu(\bx_t) - z }{b_k}}. 
	\end{align*}
	First, according to Assumption \ref{assump:F}, the second term on the right hand side of \eqref{eqn:lemma-F(1)-1} can be bounded by 
	\begin{equation} \label{result-p1}
		\sup_{z \in \calS_{\eps}, \mu \in \mathcal{N}_k} \abs{ F^{(1)}(z; \mu) - F^{(1)} (z)} \leq \ell_F r_{n_{k, exp}}.
	\end{equation}
	Next we proceed to consider the first term in \eqref{eqn:lemma-F(1)-1}. Observe that
	\begin{align*}
		& \sup_{z \in \calS_{\eps}, \mu \in \mathcal{N}_k} \abs{ \hat{F}^{(1)}(z; \mu) - F^{(1)}(z; \mu)} \\
		& \le \sup_{u \in \calS_{\eps}, \mu \in \mathcal{N}_k} \frac{1}{\hat \xi_k^2}  \Big|(\hat a_k^{(1)}-a^{(1)}) \hat \xi
		+ a^{(1)}(\hat \xi_k-\xi) - (\hat a_k - a)\hat \xi_k^{(1)} -  a (\hat \xi_k^{(1)} - \xi^{(1)}) \Big| \\
		&\quad   + \sup_{u \in \calS_{\eps}, \mu \in \mathcal{N}_k} \abs{ \frac{\xi^2-\hat \xi_k^2}{\xi^2\hat \xi_k^2}}\abs{a^{(1)}\xi-a\xi^{(1)}},
	\end{align*}
	where we have omitted the dependence of the functions on $z$ and $\mu$ for notational simplicity. By Assumption \ref{assump:xi} and Lemmas A.3 and A.4 in \cite{fan2022policy}, we obtain that with probability at least $1 - \delta$, 
	\begin{align} \label{result-p2}
		\sup_{z \in \calS_{\eps}, \mu \in \mathcal{N}_k} \abs{ \hat{F}^{(1)}(z; \mu) - F^{(1)}(z; \mu)} 
		\le C n_{k, exp}^{- \frac{m-1}{2m + 1}} \sqrt{\log n_{k, exp} } ( \sqrt d + \sqrt{\log (1/\delta)} ).
	\end{align}
	Consequently, the desired result follows from a combination of \eqref{eqn:lemma-F(1)-1}, \eqref{result-p1} and \eqref{result-p2}. This completes the proof of Lemma \ref{lemma-F(1)}.

	\subsection{Proof of Lemma \ref{lemma-sup-h}} \label{sec:lemma11-pf}
	Observe that the data-driven optimal price $ P_t = \hat{h}_{k}\circ \hat{\mu}_k(\bx_t)$.
	We begin with presenting Lemma \ref{thm:g*(x)-bound} that provides the range of $ \mu ( {\bx}_t) $ for any $\mu \in\calN_k$ and $ {\bx}_t\in\calX$, which will be useful in our proof. 
	
	\begin{lemma} \label{thm:g*(x)-bound}
		Under the conditions of Lemma \ref{lemma-sup-h}, for any $ {\bx}_t\in\calX$ and $ \mu \in\calN_k$, $\mu( {\bx}_t)\in [B_{\eps}, B-B_{\eps}]$. 
	\end{lemma}
	
	The proof of Lemma \ref{thm:g*(x)-bound} is postponed to Section \ref{lemma12-pf} in the Appendix. 
	
	\begin{proof}[Proof of Lemma \ref{lemma-sup-h}]
		Recall the definitions of estimator $\hat{h}_k(v)$ in \eqref{eqn:func-h-hat} and the population quantity $h(v)$ in \eqref{eqn:func-h} that
		\begin{equation} \label{eq-h_k}
			\hat{h}_k(v) = v + \hat{\phi}_k^{-1}(-v), 
			\quad \mbox{where}\quad \hat{\phi}_k(z) = z - \frac{1 - \hat{F}_k(z; \hat{\mu}_k)}{\hat{F}_k^{(1)}(z; \hat{\mu}_k)},
		\end{equation}
		and 
		\begin{equation} \label{eq-h}
			h(u) = v + \phi^{-1} (-v),
			\quad \mbox{where}\quad \phi(z) = z - \frac{ 1 - F(z) }{F^{(1)} (z)}. 
		\end{equation}
		Observe that 
		\begin{equation}
			\sup_{v \in [B_{\eps}, B - B_{\eps}]} | \hat{h}_k (v) - h(v) | = \sup_{v \in [B_{\eps} - B, - B_{\eps}]} | \hat{\phi}^{-1}_k( v) - \phi^{-1} (v)|.
		\end{equation}
		Thus it suffices to study the upper bound for $\sup_{v \in [B_{\eps} - B, - B_{\eps}]} | \hat{\phi}^{-1}_k( v) - \phi^{-1} (v)| $.
		
		\noindent
		{\sc Step. I.}
		First, we derive a uniform upper bound for $\abs{\hat{\phi}_k(z) - \phi(z)}$, which can be decomposed by
		\begin{equation*}
			\begin{aligned}
				\underset{z\in\calS_{\eps}}{\sup}\; \abs{\hat{\phi}_k(z) - \phi(z)}
				& \le \underset{u\in\calS_{\eps}}{\sup}\abs{\frac{\paren{1 - \hat{F}_k(z; \hat{\mu}_k)}\paren{\hat{F}_k^{(1)}(z; \hat{\mu}_k)-F^{(1)} (z)}}{\hat{F}_k^{(1)}(z; \hat{\mu}_k) F^{(1)}  (z)}} \\
				&+ \underset{z\in \calS_{\eps} }{\sup}\abs{\frac{\hat{F}_k(z; \hat{\mu}_k)-F(z)}{F^{(1)} (z)}}.
			\end{aligned}
		\end{equation*}
		By Assumption \ref{assump:F}, we have (noting that $\phi$ is an increasing function by Assumption \ref{assump:regu-phi})
		\begin{equation*}
			\min_{z \in [\phi^{-1} (B_\eps - B),  \phi^{-1} (- B_{\eps}) ]} F^{(1)} (z) \geq c_F
		\end{equation*}
		and applying the continuity of $F^{(1)} (u)$ yields that for some constant $\delta_F > 0$,
		\begin{equation*}
			\min_{z \in [\phi^{-1} (B_\eps - B) - \delta_F,  \phi^{-1} (- B_{\eps}) + \delta_F ]} F^{(1)} (z) \geq c_F / 2 
		\end{equation*}

		Therefore, it follows from Lemma \ref{lemma-F(1)} that for $\hat{\mu}_k \in \mathcal{N}_k$, with probability at least $1 - 4 \delta$, 
		\begin{equation*}
			\min_{z \in [\phi^{-1} (B_\eps - B) - \delta_{F},  \phi^{-1} (- B_{\eps}) + \delta_F]} \hat{F}^{(1)}_k (z; \hat\mu_k) \geq c_F / 4,
		\end{equation*}
		as long as $ n_{k, exp}^{\frac{m  -1} {2m + 1}} \geq (8c_F)^{-1} B_f \sqrt{\log n_{k, exp} } (\sqrt d + \sqrt {\log \delta^{-1}}) $ and $r_{n_{k, exp}} \leq c_F B_f^{-1} / 8 $.

		Combining the above results with Lemmas \ref{lemma-F} and \ref{lemma-F(1)} that bound $\abs{\hat{F}_k(z; \hat{\mu}_k)-F(z)}$ and $\abs{\hat{F}_k^{(1)}(z; \hat{\mu}_k)-F^{(1)} (z)}$, 
		we have with probability $1 - 6 \delta$,
		
		\begin{align}
			\label{phi-bound}
			\underset{z\in [\phi^{-1} (B_\eps - B) - \delta_{F},  \phi^{-1} (- B_{\eps}) + \delta_F]}{\sup}\; \abs{\hat{\phi}_k(z) - \phi(z)}
			& \le   \frac{16 B_f}{c_F^2} \cdot \paren{n_{k, exp}^{- \frac{m-1}{2m + 1}} \sqrt{\log n_{k, exp} } ( \sqrt{d} + \sqrt{\log (1/\delta)} ) + r_{n_{k, exp}}} \nonumber \\
			& + \frac{4 B_F}{c_F} \paren{n_{k, exp}^{- \frac{m}{2m + 1}} \sqrt{\log n_{k, exp} } ( \sqrt{d} + \sqrt{\log (1/\delta)} ) + r_{n_{k, exp}}}. 
		\end{align}
		
		\noindent
		{ \sc Step. II.} 
		We bound $\sup_{v \in [B_{\eps} - B, - B_{\eps}]} | \hat{\phi}^{-1}_k( v) - \phi^{-1} (v)| $ through $\underset{z\in [\phi^{-1} (B_\eps - B) - \delta_{F},  \phi^{-1} (- B_{\eps}) + \delta_F]}{\sup}\; \abs{\hat{\phi}_k(z) - \phi(z)}$. 
		The computation here involves obtaining the inverse of $\hat{\phi}_k(z)$, which is not necessarily monotone. 
		To deal with this challenge, we will show in the sequel that $\hat{\phi}_k(z)$ is very `close' to $\phi$ in some main interval of interest, which contains $[\phi^{-1} (B_\eps - B) - \delta_{F},  \phi^{-1} (- B_{\eps}) + \delta_F]$ and depends only on $F$. Since $\phi^{(1)} (z) \geq L_{\phi} > 0$, it holds that $\phi (\phi^{-1} (B_\eps - B) - \delta_{F}) \leq B_\eps - B - L_\phi \delta_F$ and similarly, $\phi (\phi^{-1} (- B_{\eps}) + \delta_F) \geq - B_{\eps} + L_\phi \delta_F$. Therefore, we have $[B_{\eps} - B, - B_{\eps}] \subset [\phi (\phi^{-1} (B_\eps - B) - \delta_{F}), \phi (\phi^{-1} (- B_{\eps}) + \delta_F)]$. Moreover, by \eqref{phi-bound}, we can deduce that as $ n_{k, exp}^{\frac{m -1} {2m+1}} \gtrsim \sqrt {\log n_{k, exp}} (\sqrt d + \sqrt {\log \delta^{-1}})$ and $r_{n_{k, exp}} \lesssim 1$, 
		\begin{equation*}
			\hat\phi_k \big( \phi^{-1} (B_\eps - B) - \delta_{F} \big) \leq B_\eps - B  
		\end{equation*}
		and 
		\begin{equation*}
			\hat\phi_k \big(\phi^{-1} (- B_{\eps}) + \delta_F \big) \geq - B_{\eps}.
		\end{equation*}
		By continuity of $\phi$ and $\hat{\phi}_k$, it follows that 
		$$
		[B_\eps - B, B_{\eps}] \subset  \phi\big([\phi^{-1} (B_\eps - B) - \delta_{F}, \phi^{-1} (- B_{\eps}) + \delta_F]\big) \cap \hat\phi_k\big([\phi^{-1} (B_\eps - B) - \delta_{F}, \phi^{-1} (- B_{\eps}) + \delta_F]\big).
		$$
		For $ v \in [B_{\eps} - B, - B_{\eps}]$, we define the inverse of $\hat{\phi}_k$ as  
		\begin{equation} \label{eqn:inv-phi}
			\hat\phi^{-1}_k(v) \defeq \inf\braces{u\in [\phi^{-1} (B_\eps - B) - \delta_{F},  \phi^{-1} (- B_{\eps}) + \delta_F]: \hat{\phi}_k(u) = v}.
		\end{equation}
		
		Now we are ready to upper bound $\underset{v\in[B_{\eps} - B, - B_{\eps}]}{\sup}\; \abs{\hat{\phi}^{-1}_k(v) - \phi^{-1}(v)}$. 
		From Assumption \ref{assump:regu-phi} that $\phi^{(1)} (u) \geq L_{\phi}$,  it holds that $\abs{\phi(u_1) - \phi(u_2)} \geq L_{\phi} \abs{u_1 -  u_2}  $ and hence with probability at least $1-6\alpha$,
		\begin{equation}  \label{eq-phi-inverse}
			\begin{aligned}
				\sup_{v \in [B_{\eps} - B, - B_{\eps}]} \abs{ \hat{\phi}_k^{-1} (v) - \phi^{-1} (v)} & \le L_{\phi}^{-1} \sup_{v \in [B_{\eps} - B, - B_{\eps}]} \abs{\phi( \hat{\phi}_k^{-1} (v) ) - \phi(\phi^{-1} (v))} \\
				& = L_{\phi}^{-1} \sup_{v \in [B_{\eps} - B, - B_{\eps}]} \abs{  \phi( \hat{\phi}_k^{-1} (v) )  - \hat{\phi}_k( \hat{\phi}_k^{-1} (v) ) } \\
				& \le L_{\phi}^{-1} \sup_{u\in  [\phi^{-1} (B_\eps - B) - \delta_{F},  \phi^{-1} (- B_{\eps}) + \delta_F] }\abs{\hat{\phi}_k(u) - \phi(u)} \\
				& \le  L_{\phi}^{-1}\frac{16 B_f}{c_F^2} \cdot \paren{n_{k, exp}^{- \frac{m-1}{2m + 1}} \sqrt{\log n_{k, exp} } ( \sqrt{d} + \sqrt{\log (1/\alpha)} ) + r_{n_{k, exp}}} \\
				& + L_{\phi}^{-1} \frac{4B_F}{c_F} \paren{n_{k, exp}^{- \frac{m}{2m + 1}} \sqrt{\log n_{k, exp} } ( \sqrt{d} + \sqrt{\log (1/\alpha)} ) + r_{n_{k, exp}}}.
			\end{aligned}
		\end{equation}
		This completes the proof of Lemma \ref{lemma-sup-h}. 
	\end{proof}

	\subsection{Proof of Lemma $\ref{thm:g*(x)-bound}$} \label{lemma12-pf}
	
	\begin{proof}[Proof of Lemma $\ref{thm:g*(x)-bound}$]
		First, we show that for any $\bx_t \in \mathcal{X}$,
		\begin{equation}\label{eqn:g-1}
			\mathring{\mu}( {\bx}_t) \in [B_v+B_{\eps}, B-B_v-B_{\eps}]. 
		\end{equation}
		This can be validated as follows.
		Recall the definition that $v_t = \mathring{\mu}( {\bx}_t) + \eps_t$, where by assumption $\eps_t\in[-B_{\eps}, B_{\eps}]$ and $v_t \in [B_v, B-B_v]$. 
		Since ${\bx}_t$ is independent of $\eps_t$, to ensure $v_t \in [0, B ]$, it must follow that  $\mathring{\mu}( {\bx}_t) \in [ B_v + B_{\eps}, B - B_v -B_{\eps}]$. 
		
		At the same time, we have 
		\begin{equation}\label{eqn:g-2}
			\underset{ {\bx}_t\in\calX, \mu \in\calN_k}{\sup}\abs{\mu( {\bx}_t)-\mathring{\mu}( {\bx}_t)}
			\le r_{n_{k, exp}} \le B_v, 
		\end{equation}
		where the last inequality holds when $n_{k, exp}$ is sufficiently large. 
		Combining \eqref{eqn:g-1} and \eqref{eqn:g-2}, we obtain the desired result. 
	\end{proof}

	\section{Proofs of Uniform Finite-Sample Bounds of DNN or TDNN} \label{append:proof-dnn-tdnn}
	Recall that the DNN and TDNN estimators are constructed using the sample $\bz_t = \{(\bx_t, g_t) \}_{t \in \cT_{k, exp}}$ of size $n_{k,exp}$, where $g_t = B y_t$ with $y_t$ being an Bernoulli random variable.
	To simplify the notation, we denote $n = n_{k, exp}$ in this section when there is no ambiguity. 
	
	\subsection{Proof of Theorem \ref{le-DNN-TDNN}}
	
	We first prove the result \eqref{re-DNN} for the DNN estimator. The main idea of proof is to apply a bias-variance decomposition for the DNN estimator. The bias term can be bounded using existing result in \cite{demirkaya2022optimal} and the variance term can be addressed by applying the bound differences inequality in empirical theory and concentration inequality for U-statistic.
	Observe that the estimation error can be decomposed as 
	\begin{equation} \label{decomp}
		\begin{aligned}
			\hat\mu_k^{\DNN}(\bx; s) - \mathring{\mu}(\bx) = \paren{  \hat\mu_k^{\DNN}(\bx; s) - \mathbb{E} g_{(1)} (\bz_1, \cdots, \bz_s) } 
			+ \paren{ \mathbb{E} g_{(1)} (\bz_1, \cdots, \bz_s) - \mathring{\mu}(\bx) }. 
		\end{aligned}
	\end{equation}
	For the second term related to the bias, it follows from applying a similar technique of proving Theorem 1 in \cite{demirkaya2022optimal} that 
	\begin{equation} \label{eq_bias_DNN}
		\sup_{\bx \in \mathcal{X} } \big| \mathbb{E} g_{(1)} (\bz_1, \cdots, \bz_s) - \mathring{\mu}(\bx) \big| \leq C s^{-2/d},
	\end{equation}
	where $C$ is a constant depending on the underlying density function $f(\cdot)$ of $\bx_t$,  the mean utility function $\mathring{\mu} (\cdot)$ and the dimensionality $d$. 
	
	For the first term related to the variance in \eqref{decomp}, observe that $\hat{\mu}_k^{\DNN}(\bx; s)$ is a U-statistic of order $s$. The main idea is to apply the concentration inequality for U-statistics and the fact that there are only finite number of distinct $K$-NN sets over $\bx \in \mathcal{X}$ for any $K$, given $n$ i.i.d samples. 
	First, by Hoeffding inequality of U-statistic (\cite{ai2022hoeffding}), noting that $0 \leq g_t := B y_t \leq B $, we have for any $t > 0$.
	\begin{equation} \label{eq_U_concen}
		\mathbb{P} \big( \big|  \hat{\mu}_k^{\DNN}(\bx; s) - \mathbb{E} g_{(1)} (\bz_1, \cdots, \bz_s) \big| \geq t \big) \leq 2 \exp \Big( - \frac{n t^2} {2 B^2 s } \Big).
	\end{equation}
	\ignore{
		Therefore, we can obtain that 
		\begin{equation} 
			\mathbb{E} \exp \Big\{ \big| \hat{\mu}_k^{\DNN}(\bx; s) - \mathbb{E} g_{(1)} (\bz_1, \cdots, \bz_s) \big|^2 \cdot \frac{ n } {8 B^2 s } \Big\}  \leq 2.
		\end{equation}
	}
	In addition, note that Lemma 3 in \cite{jiang2019non} shows that there are at most $d n^d$ distinct $K$-NN sets over $\bx 
	\in \mathcal{X} \subset \mathbb{R}^d$ for any $K$, given $n$ i.i.d. samples. Therefore, for fixed $\bz_1, \cdots, \bz_n$ and $s$, the DNN estimator $\hat{\mu}_k^{\DNN}(\bx; s)$ can take at most $d n^d$ distinct possible values in view of its L-statistic representation in \eqref{eqn:dnn-L-stat}. Consequently, applying union bound and \eqref{eq_U_concen} yields for any $t > 0$, 
	\begin{equation}
		\begin{aligned}
			\mathbb{P} \Big( \sup_{\bx \in \mathcal{X}} \big|  \hat{\mu}_k^{\DNN}(\bx; s) - \mathbb{E} g_{(1)} (\bz_1, \cdots, \bz_s) \big| \geq t \Big) \leq 2 d n^{d} \exp \Big( - \frac{n t^2} {2 B^2 s } \Big),
		\end{aligned}
	\end{equation}
	Therefore, we have 
	\begin{equation} \label{eq_DNN_variance}
		\mathbb{P} \bigg( \sup_{\bx \in \mathcal{X}} \big|  \hat{\mu}_k^{\DNN}(\bx; s) - \mathbb{E} g_{(1)} (\bz_1, \cdots, \bz_s) \big| \geq B \sqrt{ \frac{2 s [\log \delta^{-1}  + \log d + d \log n ] } {n}}  \bigg) \leq  2\delta ,
	\end{equation}
	which combining with \eqref{eq_bias_DNN} derives the desired uniform convergence rate for DNN estimator given in \eqref{re-DNN}.

	Next we deal with the TDNN estimator $\hat{\mu}_k^{\TDNN}(\bx; s_1, s_2)$. Recall that the TDNN estimator is a linear combination of two DNN estimators with distinct subsampling scales $s_1$ and $s_2$, that is,
	\begin{equation*}
		\hat{\mu}_k^{\TDNN}(\bx; s_1, s_2) = \alpha_1 \hat{\mu}_k^{\DNN}(\bx; s_1)  + \alpha_2  \hat{\mu}_k^{\DNN}(\bx; s_2),  
	\end{equation*}
	where $\alpha_1= 1/(1 - (s_1/s_2)^{-2/d})$ and $\alpha_2  = - (s_1 / s_2)^{-2/d} / (1 - (s_1/s_2)^{-2/d})$. Similarly, we have the bias-variance decomposition
	\begin{equation} \label{decomp-TDNN}
		\begin{aligned}
			\hat\mu_k^{\TDNN}(\bx; s_1, s_2) - \mathring{\mu}(\bx) & = \paren{  \hat\mu_k^{\TDNN}(\bx; s_1, s_2) - \alpha_1 \mathbb{E} g_{(1)} (\bz_1, \cdots, \bz_{s_1})  - \alpha_2 \mathbb{E} g_{(1)} (\bz_1, \cdots, \bz_{s_2}) }\\
			&+ \paren{ \alpha_1 \mathbb{E} g_{(1)} (\bz_1, \cdots, \bz_{s_1}) + \alpha_2 \mathbb{E} g_{(1)} (\bz_1, \cdots, \bz_{s_2}) - \mathring{\mu}(\bx) }. 
		\end{aligned}
	\end{equation}
	By Corollary 1 in \cite{demirkaya2022optimal}, it holds that 
	\begin{equation} \label{TDNN-bias}
		\sup_{\bx \in \mathcal{X}}  \big|  \alpha_1 \mathbb{E} g_{(1)} (\bz_1, \cdots, \bz_{s_1}) + \alpha_2 \mathbb{E} g_{(1)} (\bz_1, \cdots, \bz_{s_2}) - \mathring{\mu}(\bx)  \big| \leq C s^{- \min\{3,\, 4/d\} },
	\end{equation}
	where $C$ is a constant depending on the underlying density function $f(\cdot)$ of $\bx_t$,  the mean utility function $\mathring{\mu} (\cdot)$ and the dimensionality $d$. 
	
	Moreover, we have for the variance term that 
	\begin{equation*}
		\begin{aligned}
			&  \sup_{\bx \in \mathcal{X}}  \big| \hat\mu_k^{\TDNN}(\bx; s_1, s_2) - \alpha_1 \mathbb{E} g_{(1)} (\bz_1, \cdots, \bz_{s_1})  - \alpha_2 \mathbb{E} g_{(1)} (\bz_1, \cdots, \bz_{s_2}) \big| \\
			& \leq |\alpha_1|  \sup_{\bx \in \mathcal{X}}  \big| \hat\mu_k^{\DNN}(\bx; s_1) -   \mathbb{E} g_{(1)} (\bz_1, \cdots, \bz_{s_1}) \big|   + |\alpha_2 |  \sup_{\bx \in \mathcal{X}}  \big| \hat\mu_k^{\DNN}(\bx; s_2) -   \mathbb{E} g_{(1)} (\bz_1, \cdots, \bz_{s_2}) \big|.
		\end{aligned}
	\end{equation*}
	Therefore, it follows from \eqref{eq_DNN_variance} that with probability $1 - 4 \delta$,
	\begin{equation} \label{TDNN-var}
		\begin{aligned}
			& \sup_{\bx \in \mathcal{X}}  \big| \hat\mu_k^{\TDNN}(\bx; s_1, s_2) - \alpha_1 \mathbb{E} g_{(1)} (\bz_1, \cdots, \bz_{s_1})  - \alpha_2 \mathbb{E} g_{(1)} (\bz_1, \cdots, \bz_{s_2}) \big| \\
			& \leq  B \big(|\alpha_1| \sqrt {s_1} + |\alpha_2| \sqrt {s_2} \big) \sqrt{ \frac{2 [\log \delta^{-1}  + \log d + d \log n ] } {n}} \\
			& \leq C_d B \sqrt{ \frac{2 s_1 [\log \delta^{-1}  + \log d + d \log n ] } {n}} 
		\end{aligned}
	\end{equation}
	where $C_d$ is a constant depending on the dimensionality $d$ and the ratio of $s_1/s_2$. Substituting \eqref{TDNN-bias} and \eqref{TDNN-var} into \eqref{decomp-TDNN} yields the desired result \eqref{re-TDNN} for TDNN estimator. The proof of Theorem \ref{le-DNN-TDNN} is completed.
\end{appendices}

\end{document}